\documentclass[10pt, journal,twocolumn]{IEEEtran}
\pdfoutput=1

\usepackage{amsmath}
\usepackage{mathtools}
\usepackage{amsmath}
\usepackage{amssymb}
\usepackage{cite}
\usepackage{array}
\usepackage{amsthm}
\usepackage{amsmath}
\usepackage{tabulary}
\usepackage{multirow}
\usepackage{hyperref}
\usepackage{subfig}
\IEEEoverridecommandlockouts
\usepackage{pifont}
\usepackage{epsfig}
\usepackage{graphicx}
\usepackage{url}
\usepackage{amssymb}

\graphicspath{ {./figs/} }
\newcounter{MYtempeqncnt}

\newcommand{\expect}[1]{{\mathbb{E}\left[{#1}\right]}}
\newcommand{\cexpect}[2]{{\mathbb{E}_{#2}\left[{#1}\right]}}
\newcommand{\pr}{{\mathbb{P}}}

\newcommand{\mgf}{\mathcal{L}}
\newcommand{\pmf}{{p}}
 
\newcommand{\set}[1]{{\mathcal{#1}}}

\newcommand{\ple}[1]{{\alpha_{#1}}}
\newcommand{\nple}[1]{{\hat{\alpha}_{#1}}}
\newcommand{\power}[1]{\mathrm{P}_{#1}}
\newcommand{\res}{\mathrm{W}}
\newcommand{\npower}[1]{\mathrm{\hat{P}}_{#1}}
\newcommand{\bias}[1]{\mathrm{B}_{#1}}
\newcommand{\noisepower}{\sigma^2}
\newcommand{\nbias}[1]{\mathrm{\hat{B}}_{#1}}

\newcommand{\dnsty}[1]{{\lambda_{#1}}}

\newcommand{\userdnsty}{\lambda_u}
\newcommand{\R}{{\mathbb{R}}}
\newcommand{\SINRthresh}{\tau} 
\newcommand{\Z}{\mathsf{Z}} 
\newcommand{\Q}{\mathsf{Q}}

\newcommand{\bkhl}[1]{\mathrm{O}_{#1}}
\newcommand{\RATEthresh}{\rho}

\newcommand{\nRATEthresh}{\hat{\rho}}
 
\newcommand{\uRATEthresh}{t}

\newcommand{\uth}[1]{#1^\text{th}} 
\newcommand{\uset}{\mathcal{U}}
\newcommand{\opueindex}{B}
\newcommand{\ipueindex}{\bar{B}}
\newcommand{\mue}{\mathcal{U}_1}
\newcommand{\pue}{\mathcal{U}_2}
\newcommand{\opue}{\mathcal{U}_{\opueindex}}
\newcommand{\ipue}{\mathcal{U}_{\ipueindex}}
\newcommand{\tmap}[1]{J(#1)}

\newcommand{\SINR}{\mathtt{SINR}}

\newcommand{\SNR}{\mathtt{SNR}}
\newcommand{\assocr}{\mathcal{C}}
\newcommand{\rate}[1]{R_{#1}}
\newcommand{\ndist}[1]{z_{#1}}
\newcommand{\ndistns}{z}
\newcommand{\NDIST}[1]{Z_{#1}}
\newcommand{\NAP}[1]{X_{#1}^*}

\newcommand{\ndistnsc}{y}
\newcommand{\NDISTc}[1]{Y_{#1}}
\newcommand{\PPP}[1]{\Phi_{#1}}

\newcommand{\PPPu}{\Phi_{u}}
\newcommand{\tierPPP}[1]{{\Phi}_{#1}}
\newcommand{\load}[1]{{N}_{#1}} 
\newcommand{\oload}[1]{{N}_{o,#1}} 
 
\newcommand{\avload}[1]{\bar{N}_{#1}} 
 
\newcommand{\chanl}{H}

\newcommand{\area}{C}
\newcommand{\areans}{c}
\newcommand{\indic}{1\hspace{-2mm}{1}}

\newcommand{\pcov}{\mathcal{S}}
\newcommand{\psys}{\mathcal{R}}
\newcommand{\rcov}{\mathcal{R}}

\newcommand{\avpsys}{\bar{\mathcal{R}}}

\newcommand{\passoc}{\mathcal{A}}

\newcommand{\pa}[1]{p_{a#1}} 
\newcommand{\af}{\eta} 
\newcommand{\ap}{\delta} 
\newcommand{\effres}[1]{\gamma_{#1}}

\newtheorem{thm}{{\bf Theorem}}
\newtheorem{cor}{Corollary}
\newtheorem{rem}{Remark}
\newtheorem{lem}{Lemma}

\theoremstyle{definition}
\newtheorem{definition}{Definition}

\begin{document}
\title{Joint Resource Partitioning  and Offloading in Heterogeneous Cellular Networks}

\author{Sarabjot Singh and Jeffrey G. Andrews \thanks{This work has been supported by the Intel-Cisco Video Aware Wireless Networks (VAWN) Program and NSF grant CIF-1016649. A part of this paper
is accepted for presentation at IEEE Globecom 2013 in Altlanta, USA \cite{SinAndGC13}. 

 The authors are with Dept. of Electrical and Computer Engineering at the University of Texas, Austin  (email:  sarabjot@utexas.edu   and jandrews@ece.utexas.edu).}}
\maketitle

\begin{abstract}
In  heterogeneous cellular networks (HCNs), it is desirable to offload mobile users to small cells, which are typically significantly less congested than the macrocells.  To achieve sufficient load balancing, the offloaded users often have much lower SINR than they would on the macrocell. This SINR degradation can be partially alleviated through interference avoidance, for example time or frequency resource partitioning, whereby the macrocell turns off in some fraction of such resources.  Naturally, the optimal offloading strategy is tightly coupled with resource partitioning; the optimal amount of which in turn depends on how many users have been offloaded.  In this paper, we propose a general and tractable framework for modeling and analyzing joint resource partitioning and offloading in a two-tier cellular network. With it, we are able to derive the downlink rate distribution over the entire network,  and an optimal strategy for joint resource partitioning and offloading. We show that load balancing, by itself, is insufficient, and  resource partitioning is required  in conjunction with offloading to improve the rate of cell edge users  in co-channel heterogeneous networks.
\end{abstract}

\section{Introduction}
The exponential growth in mobile traffic primarily driven by mobile video  has led the drive to deploy low power base stations (BSs) both in  licensed  and unlicensed spectrum  in order to complement the existing macro cellular architecture. Increased cell density increases the area spectral efficiency of the network \cite{alouini1999area} and is the key factor for the very large required capacity boost. Such a heterogeneous network (HetNet) consists of macro BSs coexisting with small cells formed by co-channel low power base stations like micro, pico and femto BSs, as well as unlicensed band  WiFi access points (APs) \cite{qcom_hetnet_wmag,ghosh2012heterogeneous}. 

The ``natural'' association/coverage areas of the low power APs tend to be much smaller than those of the macro BSs, and hence the fraction of user population that is offloaded to small cells may often be limited, resulting in insufficient relief to the congested  macro tier. This  user load disparity not only leads to suboptimal rate distribution across the network, but the lightly loaded small cells may also  lead to  bursty interference causing QoS degradation \cite{singh2012interference,madan2010cell}. 
 One practical technique  for proactively offloading users to small cells  is called cell range expansion (CRE) \cite{qcom_hetnet_wmag} wherein the users are offloaded through an association bias. A positive association bias implies that a user would be offloaded to a small cell as soon as the received power difference from the macro and  small cell drops  below the bias value,   which ``artificially" expands the association areas of small cells\footnote{Access point (AP), base station (BS), cell are used interchangeably in the paper.}.  Though experiencing reduced congestion, in co-channel deployments such offloaded users also have degraded signal-to-interference-plus-noise-ratio ($\SINR$),  as the strongest AP (in  terms of received power)  now  contributes to interference. Therefore, the gains from balancing load could be negated if suitable interference avoidance  strategies are not adopted in conjunction with cell range  expansion particularly  in co-channel deployments \cite{madan2010cell}. One such strategy  of interference avoidance  is resource partitioning \cite{qcom_hetnet_wmag,lopez2011enhanced}, wherein the transmission of  macro tier  is periodically muted on certain fraction of  radio resources (also called \textit{almost blank  subframes} in $3$GPP LTE). The  offloaded users can then be scheduled in these resources by the small cells leading to their protection from co-channel macro tier interference.
 
\subsection{Motivation and Related Work}
It  has been established that without proactive offloading and resource partitioning only limited performance gains can be achieved from the deployment of small cells  \cite{vajapeyam2011downlink,WangPed12,3ggp_r1_103824,barbieri2012coordinated,wang2012sensitivity}. These techniques  are strongly coupled and directly influence the rate of users, but the fundamentals of jointly optimizing  offloading and resource partitioning are not well understood.  For example,   an excessively large association bias can  cause the small cells to be overly congested with users of poor $\SINR$,  which requires excessive muting by  the macro cell to improve the rate of offloaded users. 
Earlier simulation based studies \cite{barbieri2012coordinated,wang2012sensitivity} confirmed this insight and  showed that excessive biasing and resource partitioning  can actually degrade the overall rate distribution, whereas the choice of optimal parameters can yield about 2-3x gain in the rate coverage (fraction of user population receiving rate greater than a threshold). 
Although encouraging,  a general tractable framework for  characterizing the optimal operating regions for resource partitioning and offloading is still an open problem.  
The work in this  paper is aimed to bridge this gap.

A ``straightforward" approach of finding the optimal strategy is to search over all possible user-AP associations and time/frequency allocations for each network configuration. Besides being computationally daunting, this approach is unlikely to lead to insight into the role of  key parameters on  system performance. Another methodology is a probabilistic  analytical approach, where the network configuration is assumed random and following a certain distribution. This has the advantage of leading to  insights on the impact of various system parameters on the average performance through tractable expressions. Analytical approaches for  biasing and interference coordination were studied in \cite{mukherjee2011effects, novlan12analytical,lopez2012expanded}, but downlink rate (one of the key metrics) was not investigated. Optimal bias and almost blank subframes were prescribed in \cite{mukherjee2011effects} based on average per user spectral efficiency. 
 A related $\SINR$ and mean throughput  based analysis for resource partitioning was done in \cite{novlan12analytical} and \cite{CiernyABSHetNets13} respectively, but offloading was not captured.  The choice of optimal  range expansion biases in \cite{lopez2012expanded} was not based on rate distribution.  In this paper, we use the metric of \emph{rate  coverage}, which captures the effect of both $\SINR$ and load distribution across the network. Semi-analytical approaches in \cite{guvenc2011capacity,ye2012user} showed, through simulations, that there exists an optimal association  bias  for fifth percentile and median  rate which is  confirmed in this paper through our analysis.  Also, to the best of our knowledge,  none of the mentioned earlier works considered the impact of backhaul capacities on offloading, which is another contribution of the presented work.

\subsection{Approach and Contributions}
We propose a general and tractable  framework to analyze joint resource partitioning and offloading in a two-tier cellular network in Section  \ref{sec:sysmodel}. The proposed modeling can be extended to a multiple tier setting as discussed in Sec. \ref{sec:multier}. Each  tier of base stations  is modeled as  an independent  Poisson point process (PPP), where each tier differs in transmit power, path loss exponent, and  deployment density. The  mobile user locations are modeled as an independent PPP and user association is assumed to be based on biased received power.  On all channels,  i.i.d. Rayleigh fading is assumed.
Similar tractable frameworks were used for deriving $\SINR$ distribution in HCNs in \cite{josanxiaand12,dhiganbacand12,mukh12}.  The empirical validation in \cite{andganbac11} and theoretical validation in \cite{BlaKarKee12} for heavily shadowed cellular networks have strengthened the case of modeling macro cellular networks using a  PPP.  Due to the formation of random association/coverage areas in such network models,  load distribution is difficult  to characterize.  An approximate load and rate distribution was derived for  multiple radio access technology (RAT) HetNets in \cite{SinDhiAnd13}.

Based on  our proposed approach, the contributions of the paper can be divided into two categories:\\
\textbf{Analysis.} 
The rate complementary cumulative distribution function (CCDF) in a two-tier co-channel  heterogeneous network is derived as a function of the cell range expansion/offloading and resource partitioning parameters in Section \ref{sec:ratedist}. Rate coverage at a particular rate threshold is the rate  CCDF value at that threshold.  The derived  rate distribution is then modified to incorporate a network setting where APs are equipped with limited capacity backhaul.   Under certain plausible scenarios, the derived expressions are in closed form.\\
\textbf{Design Guidelines.}
The theoretical results lead to joint resource partitioning and offloading insights for optimal $\SINR$  and rate coverage in Section \ref{sec:optbias}. In particular, we show the following: 
\begin{itemize}
\item With no resource partitioning, optimal association bias  for rate coverage is independent of the density of the small cells. In contrast, offloading is shown to be strictly suboptimal for $\SINR$ in this case.
\item With resource partitioning, optimal association bias decreases with increasing density of the small cells.
\item In both of the above scenarios, the optimal fraction of users offloaded, however, increases with increasing density of small cells.
\item With decrease in backhaul capacity/bandwidth the optimal association bias for the corresponding tier always decreases. However, in contrast to the trend in the ``infinite"\footnote{Infinite bandwidth implies sufficiently large so as not to affect the effective end-to-end rate.} backhaul scenario, the optimal association bias may increase with increasing small cell  density.
\end{itemize}
The paper is concluded  in Section \ref{sec:conclusion} and future work  is suggested.

\section{Downlink System Model and Key Metrics}\label{sec:sysmodel}
In this paper, the wireless network   consists of  a two-tier deployment of APs.  The location of the APs of  $k^{\mathrm{th}}$ tier ($k=1,2$)  is  modeled as a two-dimensional homogeneous PPP $\PPP{k}$  of density (intensity) $\dnsty{k}$.  Without any loss of generality, let the  macro tier be tier $1$ and  the small cells constitute tier $2$. The locations of users (denoted by $\uset$) in the network  are modeled as another  independent homogeneous PPP $\PPPu$ with density $\userdnsty$. 
Every AP of $\uth{k}$ tier transmits with the same transmit power $\power{k}$ over bandwidth $\res$. The downlink desired and interference signals  from an AP of tier-$k$ are assumed to experience path loss with a path loss exponent $\ple{k}$. A user receives a power  $\power{k} \chanl_x x^{-\ple{k}}$ from an AP  of $\uth{k}$ tier at a distance $x$, where $\chanl_x$ is the random channel power gain.  The random channel gains are  assumed to be  Rayleigh distributed with average unit power, i.e., $\chanl_x \sim \exp(1)$.  General fading distributions can be considered at some loss of tractability \cite{BacBlaMuh09}. The noise is assumed additive with power $\noisepower$.  The notations used in this paper  are summarized in Table \ref{table:notationtable}.

\begin{table}
	\centering
\caption{Notation Summary}
	\label{table:notationtable}
  \begin{tabulary}{\columnwidth}{ |c | C | }
    \hline
    \textbf{Notation} & \textbf{Description} \\ \hline
$\PPP{k};\PPPu$  & PPP of   APs of $\uth{k}$ tier;  PPP of  mobile users  \\ \hline
$\dnsty{k};  \userdnsty$ & Density of  APs of $\uth{k}$ tier; density of mobile users \\ \hline
$\power{k}; \npower{k}$ & Transmit power of APs of $\uth{k}$ tier; normalized transmit power of APs of $\uth{k}$ tier \\\hline
$\bias{k}; \nbias{k}$ & Association bias for $\uth{k}$ tier; normalized association bias  for $\uth{k}$ tier.\\\hline
$\ple{k}; \nple{k}$ & Path loss exponent of $k^{\text{th}}$ tier; normalized path loss exponent of $\uth{k}$ tier\\\hline
$\res{}; \bkhl{k}$ & Air interface bandwidth  at an AP for resource allocation; backhaul bandwidth at an AP of $\uth{k}$ tier  \\\hline
$\set{U}_l$ & Macro cell users $l=1$, small cell users (non-range expanded) $l = \ipueindex$, offloaded users $l=\opueindex$\\\hline
$\af; \effres{l}$ & Resource partitioning fraction; inverse of the effective fraction of resources available for users in $\set{U}_l$\\\hline
$\tmap{l}$ & Map from user set index to serving tier index, $\tmap{1}=1$, $\tmap{\ipueindex}=\tmap{\opueindex}=2$\\\hline
$\noisepower$ & Thermal noise power \\\hline
$\passoc_l $& Association probability of a typical user to $\set{U}_l$\\\hline
$\rcov; \pcov; \RATEthresh$ & Rate coverage; $\SINR$ coverage; rate threshold\\\hline
$\load{l}; \pmf_l(n)$ & Load at tagged AP of $u \in \uset_l$; PMF of load $\pmf_l(n)=\pr(\load{l}=n)$\\\hline
$ \NDIST{k}; \NDISTc{l}$ & Distance of the nearest AP in $\uth{k}$ tier; distance of the tagged AP conditioned on $u\in \uset_l$\\\hline
$ \assocr_{x_k}; \area_k$ & Association region; area of an AP of tier $k$\\\hline
\end{tabulary}
\end{table}

\subsection {User Association}
The analysis in this paper is done for a \textit{typical} user $u$ located at the origin. This is allowed by Slivnyak's theorem \cite{mecke_book}, which states that the properties observed by a typical\footnote{The term typical and random   are interchangeably used in this paper.} point of a PPP  $\Phi$  is same as those observed by a node at origin  in the process $\Phi\cup\{0\}$. Let $\NDIST{k}$ denote the distance of the typical user from the nearest AP of $\uth{k}$ tier. It is assumed that each user uses biased received power association in  which it associates to the nearest AP of tier $j$ if
\begin{align}\label{eq:association}
j &= \arg \max_{k\in\left\{1,2\right\}} \power{k}\bias{k}\NDIST{k}^{-\ple{k}}, 
\end{align}
where $\bias{k}$ is the association bias  for $\uth{k}$ tier. Increasing  association bias leads to the range expansion for the corresponding APs and therefore  offloading of more users to the corresponding tier.  For clarity, we define the \textit{normalized} value of a parameter  of a tier as its value divided by the value it takes for the serving tier. Thus, 
\[ \npower{k} \triangleq \frac{\power{k}}{\power{j}},\,\, \nbias{k} \triangleq \frac{\bias{k}}{\bias{j}}\,\, \text{, and }\,\, \nple{k} \triangleq \frac{\ple{k}}{\ple{j}},  \]
are respectively the normalized transmit power, association bias, and path loss exponent of tier $k$ conditioned on the user being associated with tier $j$. 
In this paper, association bias for  tier 1 (macro tier) is assumed to be unity ($\bias{1}= 0$ dB) and that of tier 2 is simply denoted by $\bias{}$, where $\bias{} \geq 0$ dB. 
In the given setup, a  user $u\in \uset$ can lie in the following three disjoint sets:
\begin{equation}\label{eq:setdef}
u \in \begin{cases}
 \mue  \text{ if $j=1$, $\power{1}\NDIST{1}^{-\ple{1}}\geq\power{2}\bias{}\NDIST{2}^{-\ple{2}} $}\\
 \ipue  \text{ if $j=2$ and   $\power{2}\NDIST{2}^{-\ple{2}}>\power{1}\NDIST{1}^{-\ple{1}}$}\\
 \opue \text{ if $j=2$ and $\power{2}\NDIST{2}^{-\ple{2}}\leq\power{1}\NDIST{1}^{-\ple{1}}<\power{2}\bias{}\NDIST{2}^{-\ple{2}}$},
 \end{cases}
\end{equation}
where $\mue\cup\opue\cup\ipue = \set{U}$ clearly. The set $\mue$ is the set of macro cell users and the set $\ipue$ is the set of unbiased small cell users. Thus, the set $\ipue$ is independent of the association bias. The users offloaded from macro cells to small cells due to cell range expansion constitute $\opue$ and  are  referred to as the \textit{range expanded users}. All the users associated with small cells are $\pue \triangleq \ipue \cup \opue$. We define a mapping  $J:\{1,\ipueindex,\opueindex\}\to \{1,2\}$ from user set index to serving tier index. Thus, from (\ref{eq:setdef}), $\tmap{1}=1$, $\tmap{\opueindex}=\tmap{\ipueindex}=2$.

The biased received power based association model described above  leads to the formation of association/coverage areas in the Euclidean plane as described below.
\begin{definition}
\textbf{Association Region}: The region of the Euclidean plane in which all users are served by an  AP is called its association region.
Mathematically, the association region of an AP of tier $j$ located at $x$ is
\begin{equation}\label{eq:assocr}
\assocr_{x_{j}}= \bigg\{ y \in \R^2: \|y-x\| \leq \left(\frac{\power{j}\bias{j}}{\power{k}\bias{k}}\right)^{1/\ple{j}}\|y-\NAP{k}(y)\|^{\nple{k}}  \forall\,\, k \bigg\},
\end{equation}
where $\NAP{k}(y) = \arg \min\limits_{x \in \PPP{k}}\|y-x\|$.
\end{definition}
The random tessellation formed by the collection $\{\assocr_{x_{j}}\}$ of association regions is  a general case of the  multiplicatively weighted Voronoi \cite[Chapter~3]{Okabe00}, which results by using the presented model  with equal path loss exponents.
\begin{figure*}
  \centering
\subfloat[Active macro tier ]
{\label{fig:unmute}\includegraphics[width=\columnwidth]{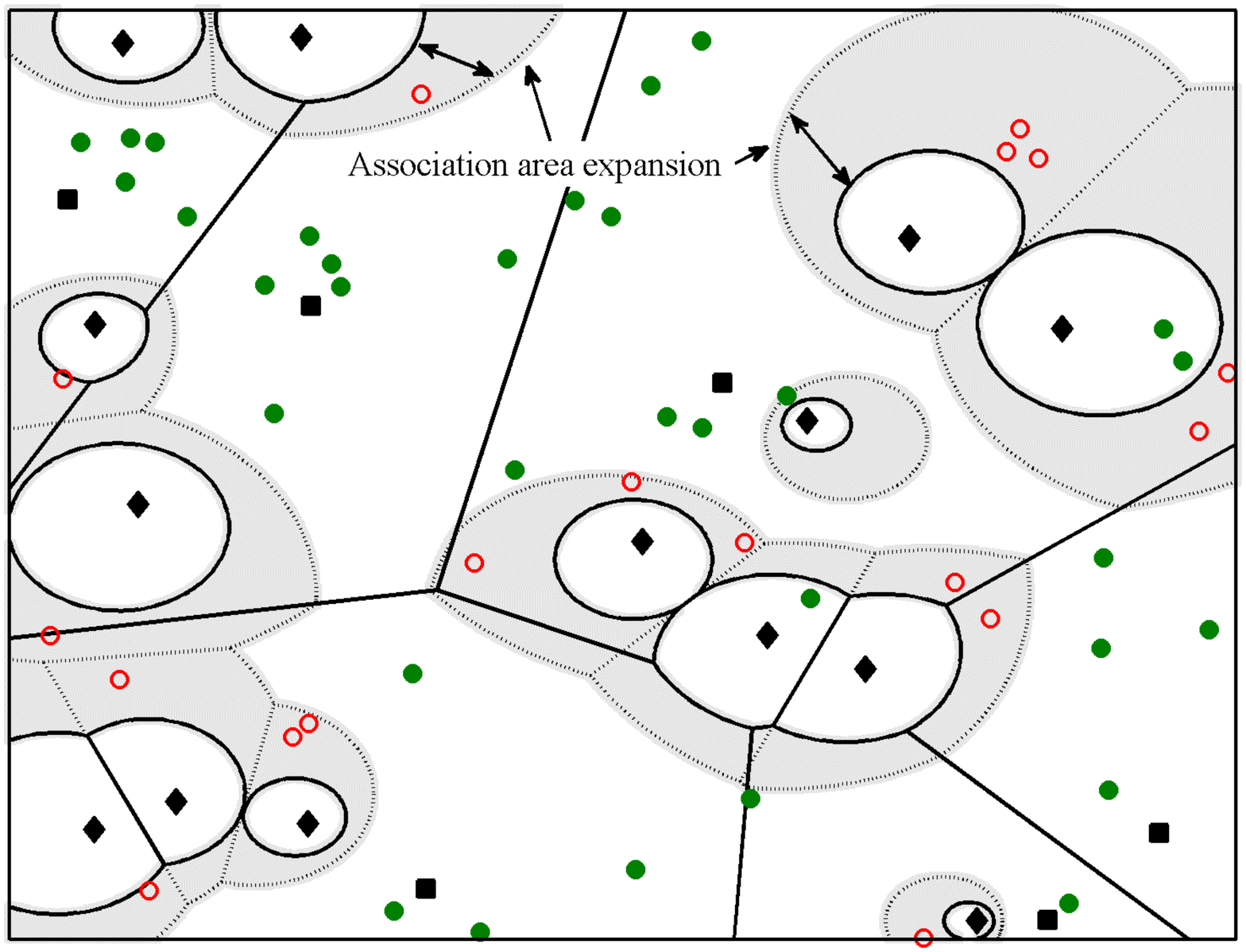}}
\subfloat[Muted macro tier]{\label{fig:mute}\includegraphics[width=\columnwidth]{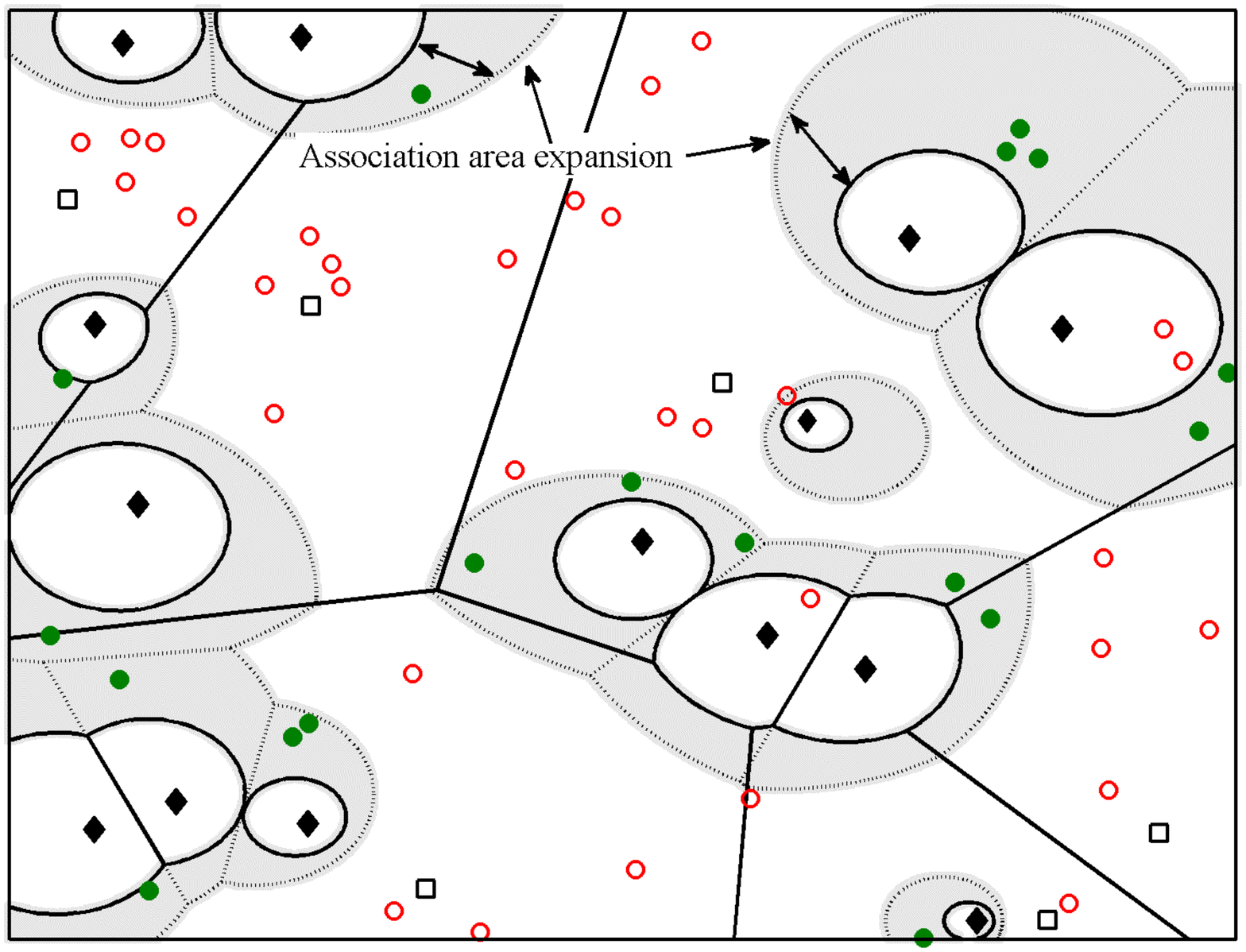}}
\caption{\small A filled marker is used for a  node  engaged in active transmission (BS) or reception (user). (a) The macro cells (filled squares)  serve the macro users $\mue$ and small cells (filled  diamonds) serve the non-range expanded users ($\ipue$) (filled circles). (b) The macro cells (hollow square) are  muted while the small cells (filled diamonds) serve the range expanded users $\opue$ (filled circles  in the shaded region).}
 \label{fig:icic}
\end{figure*}
\normalsize

\subsection{Resource Partitioning}
A resource partitioning approach is considered in which the macro cell shuts its transmission on certain fraction of time/frequency resources  and the small cell schedules the range expanded users on the corresponding resources, which  protects them  from macro cell interference.  
\begin{definition}$\boldsymbol\af$:
The resource partitioning fraction $\af$ is  the   fraction of resources on which the macro cell is inactive,  where $0 < \af <1$.  \end{definition} Thus, with resource partitioning $1-\af$ fraction of the resources at macro cell are allocated to users in $\mue$ and those at small cell are allocated to users in $\ipue$. The fraction $\af$ of the resources  in which the macro cell shuts down the transmission, the small cells schedule the range expanded users, i.e., $\opue$. Let $\effres{l}$ denote the inverse of the effective fraction of resources available for users in $\uset_l$. Then, $\effres{l} = 1/(1-\af)$ for $l \in \left\{1, \ipueindex\right\}$ and $\effres{l} = 1/\af$ for $l=\opueindex$.  The  operation  of  range expansion and  resource partitioning in a two-tier setup  is further elucidated in Fig. \ref{fig:icic}. In these plots, the  power ratio  is assumed to be $\frac{\power{1}}{\power{2}}=20$ dB and $\bias{}=10$ dB.

As a result of resource partitioning  ($0<\af<1$), the $\SINR$ of a typical user $u$, when it belongs to $\set{U}_l$, is 
\begin{equation}\label{eq:sinrdef}
\SINR = \indic\left(l\in\{1,\ipueindex\}\right)\frac{\power{\tmap{l}}\chanl_{\ndistnsc} \ndistnsc^{-\ple{\tmap{l}}}}{\sum_{k=1}^2 I_{y,k} + \noisepower}+ \indic(l=\opueindex)\frac{\power{2}\chanl_{\ndistnsc} \ndistnsc^{-\ple{2}}}{I_{y,2} + \noisepower}, 
\end{equation}
where $\indic(A)$ denotes the indicator of the event $A$, $\chanl_\ndistnsc$ is the channel power gain from the tagged  AP $s_l$ (AP serving the typical user) at a distance  $\ndistnsc$, $I_{y,k}$ denotes the interference from the $\uth{k}$ tier. The interference power from $\uth{k}$ tier is 
\begin{equation}\label{eq:interference}
I_{y,k} = \power{k}\sum_{x\in \tierPPP{k} \setminus s_l} \chanl_{x} x^{-\ple{k}}.
\end{equation}
In this paper, all APs of a tier are assumed to be active, when the corresponding tier is active. However, if each AP of $\uth{k}$ tier is independently active with a  probability $\pa{k}$, the submission in (\ref{eq:interference}) can then be treated as that over a thinned PPP of density $\dnsty{k}\pa{k}$.

 Let $\set{U}_s$ denote the set of users associated with the tagged AP. If the tagged AP belongs to macro tier, then $\load{1}=|\set{U}_s\cap\mue|$ denotes the total number of users (or \textit{load} henceforth) sharing the available  $1-\af$ fraction of the resources. Otherwise, if the tagged AP belongs to tier 2, then the load is $\load{2}=|\set{U}_s\cap\set{U}_2|$ of which $\load{\bar{B}}=|\set{U}_s\cap\ipue| $ users share the  $1-\af$ fraction of the resources and $\load{B}=|\set{U}_s\cap\opue|$ users share the rest $\af$;   $\load{2}=\load{B} + \load{\bar{B}} -1$ (one is subtracted to account for double counting of the typical user).  The available resources at an AP are assumed to be shared equally among the associated  users. This results in each user having a   rate   proportional to its link's spectral efficiency. Round-robin scheduling is an approach which results in such equipartition of resources.  Further, user queues are assumed  \textit{saturated} implying that each AP always has data to transmit to its associated mobile users. Thus, the rate of a typical user $u$ is 
\begin{equation}\label{eq:ratedef}
\rate{} = \sum_{l \in \{1,\ipueindex,\opueindex\}}\frac{\indic(u\in\set{U}_l)}{\effres{l}\load{l}}\res\log\left(1+\SINR\right).
\end{equation}
The above rate allocation model assumes infinite backhaul bandwidth for all APs, which may be particularly questionable for small cells. Discussion about limited   backhaul bandwidth is deferred to Sec. \ref{sec:bkhl}.

\subsection{Rate and $\SINR$ Coverage}
The rate and $\SINR$ coverage can be formally defined as follows.
\begin{definition} \textbf{Rate/$\SINR$ Coverage}:
The rate coverage for a rate threshold $\RATEthresh$ is  
\begin{equation}
\psys(\RATEthresh)\triangleq\pr(R>\RATEthresh),
\end{equation}
and $\SINR$ coverage for a threshold $\SINRthresh$ is  
\begin{equation}
\pcov(\SINRthresh)\triangleq \pr\left(\SINR> \SINRthresh\right). 
\end{equation}
\end{definition}
The coverage can be equivalently  interpreted  as (i) the probability that a randomly chosen user can achieve a
target threshold, (ii) the average fraction of users in the network who at any time achieve the corresponding threshold, or (iii) the average fraction of the network area that is receiving rate/$\SINR$ greater than the rate/$\SINR$ threshold. 

\section{Rate Distribution}\label{sec:ratedist}
This section derives the  load distribution and $\SINR$ distribution, which are subsequently used for deriving the  rate distribution (coverage) and is the main technical section of the paper.
\subsection{\texorpdfstring{$\SINR$}{SINR} Distribution}
For completely  characterizing the $\SINR$ and rate distribution, the average fraction of users belonging to the respective three disjoint sets ($\mue$, $\ipue$, and $\opue$) is needed.  Using the  ergodicity  of the  PPP, these fractions are equal to the association probability of a typical user to these sets, which are derived in the following lemma.
\begin{lem}\label{lem:aspr}
(Association probabilities)  The association probability, defined as $\passoc_{l} \triangleq \pr(u\in \set{U}_l)$, is given below for each set
\begin{equation}\label{eq:aspr}
\passoc_{1}= 2\pi\dnsty{1}\int_{0}^{\infty}{z \exp\left(-\pi\sum_{k=1}^2\dnsty{k}(\npower{k}\nbias{k})^{2/\ple{k}}z^{2/\nple{k}}\right)}\mathrm{d}z,
\end{equation}
\begin{equation}
\passoc_{\ipueindex} = 2\pi\dnsty{2}\int_{0}^{\infty}{z \exp\left(-\pi\sum_{k=1}^2\dnsty{k}(\npower{k})^{2/\ple{k}}z^{2/\nple{k}}\right)}\mathrm{d}z,
\end{equation}
\begin{multline}
\passoc_{\opueindex}= 2\pi\dnsty{2}\int_{0}^{\infty}z \Bigg\{\exp\left(-\pi\sum_{k=1}^2\dnsty{k}(\npower{k}\nbias{k})^{2/\ple{k}}z^{2/\nple{k}} \right)
\\- \exp\left(-\pi\sum_{k=1}^2\dnsty{k}(\npower{k})^{2/\ple{k}}z^{2/\nple{k}} \right)\Bigg\}\mathrm{d}z.
 \end{multline}
If path loss exponents are same, i.e., $\ple{k}\equiv\ple{}$, the association probabilities simplify to:
\begin{equation}\label{eq:simassocpr}
\begin{aligned}
\passoc_{1}&= \frac{\dnsty{1}}{\sum_{k=1}^2\dnsty{k}(\npower{k}\nbias{k})^{2/\ple{}}}\,,\passoc_{\ipueindex} = \frac{\dnsty{2}}{\sum_{k=1}^2\dnsty{k}(\npower{k})^{2/\ple{}}}\nonumber,
\end{aligned}
\end{equation}
\begin{equation}
\passoc_{\opueindex} =\frac{\dnsty{2}}{\sum_{k=1}^2\dnsty{k}(\npower{k}\nbias{k})^{2/\ple{}}}- \frac{\dnsty{2}}{\sum_{k=1}^2\dnsty{k}(\npower{k})^{2/\ple{}}}.
\end{equation}
\end{lem}
\begin{proof}
See Appendix \ref{sec:proofassocpr}.
\end{proof}

\begin{figure*}
\setcounter{MYtempeqncnt}{\value{equation}}
\setcounter{equation}{13}
\begin{align}
\pcov_1(\SINRthresh)& =2\pi\frac{\dnsty{1}}{\passoc_1}\int\limits_{0}^{\infty}\ndistnsc \exp\left\{-\frac{\SINRthresh}{\SNR_1(y)} -\pi\sum_{k=1}^2\dnsty{k}\npower{k}^{2/\ple{k}}\Q(\SINRthresh,\ple{k},\nbias{k})\ndistnsc^{2/\nple{k}}\right\}\mathrm{d} \ndistnsc \label{eq:pcov1}\\
\pcov_{\ipueindex}(\SINRthresh)& =2\pi\frac{\dnsty{2}}{\passoc_{\ipueindex}}\int_{0}^{\infty}\ndistnsc \exp\left\{-\frac{\SINRthresh}{\SNR_2(y)} -\pi\sum_{k=1}^2\dnsty{k}\npower{k}^{2/\ple{k}}\Q(\SINRthresh,\ple{k},1)\ndistnsc^{2/\nple{k}}\right\}\mathrm{d} \ndistnsc \label{eq:pcov2}
\end{align}
\small
\begin{align}
\pcov_{\opueindex }(\SINRthresh) =2\pi\frac{\dnsty{2}}{\passoc_{\opueindex }}&\int_{0}^{\infty}\ndistnsc \exp\left\{-\frac{\SINRthresh}{\SNR_2(y)}-\pi\dnsty{2}\Q(\SINRthresh,\ple{2},1)\ndistnsc^2-\pi\dnsty{1}\npower{1}^{2/\ple{1}}\ndistnsc^{2/\nple{1}}\right\}\left\{\exp\left(-\pi\dnsty{1}\npower{1}^{2/\ple{1}}\ndistnsc^{2/\nple{1}}(\nbias{1}^{2/\ple{1}}-1)\right)
-1\right\}\mathrm{d} \ndistnsc \label{eq:pcov3}\,,
\end{align}
\normalsize
\setcounter{equation}{\value{MYtempeqncnt}}
\hrulefill
\vspace*{4pt}
\end{figure*}
\begin{figure*}
\setcounter{MYtempeqncnt}{\value{equation}}
\setcounter{equation}{18}
\begin{align}
\psys_1(\RATEthresh) &=2\pi\frac{\dnsty{1}}{\passoc_{1}} \sum_{n\ge 1}\pmf_1(n)
\int_{0}^{\infty}\ndistnsc \exp \Bigg\{-\frac{\uRATEthresh(n\nRATEthresh \effres{1})}{\SNR_1(y)} -\pi\sum_{k=1}^2\dnsty{k}\npower{k}^{2/\ple{k}}\Q(\uRATEthresh(n\nRATEthresh \effres{1}),\ple{k},\nbias{k})\ndistnsc^{2/\nple{k}}\Bigg\}\mathrm{d} \ndistnsc \,\label{eq:rcov1}\\
\psys_{\ipueindex}(\RATEthresh) & =2\pi\frac{\dnsty{2}}{\passoc_{\ipueindex}} \sum_{n\ge 1}\pmf_{\ipueindex}(n)\int_{0}^{\infty}\ndistnsc \exp \Bigg\{-\frac{\uRATEthresh(n\nRATEthresh\effres{\ipueindex} )}{\SNR_2(y)} -\pi\sum_{k=1}^2\dnsty{k}\npower{k}^{2/\ple{k}}\Q(\uRATEthresh(n\nRATEthresh\effres{\ipueindex}),\ple{k},1)\ndistnsc^{2/\nple{k}}\Bigg\}\mathrm{d} \ndistnsc \label{eq:rcov2}\\
\psys_{\opueindex }(\RATEthresh) &=2\pi\frac{\dnsty{2}}{\passoc_{\opueindex }} \sum_{n\ge 1}\pmf_{\opueindex}(n)\int_{0}^{\infty}\ndistnsc\Bigg\{ \exp\left(-\frac{\uRATEthresh(n\nRATEthresh\effres{\opueindex })}{\SNR_2(y)} -\pi\dnsty{2}\ndistnsc^2\Q(\uRATEthresh(n\nRATEthresh\effres{\opueindex }),\ple{2},1)-\pi\dnsty{1}\ndistnsc^{2/\nple{1}}(\npower{1}\nbias{1})^{2/\ple{1}}\right)\nonumber\\&
-\exp\left(-\frac{\uRATEthresh(n\nRATEthresh \effres{\opueindex })}{\SNR_2(y)} -\pi\dnsty{2}\ndistnsc^2\Q(\uRATEthresh(n\nRATEthresh \effres{\opueindex }),\ple{2},1 )-\pi\dnsty{1}\ndistnsc^{2/\nple{1}}(\npower{1})^{2/\ple{1}}\right)\Bigg\}\mathrm{d} \ndistnsc, \label{eq:rcov3}
\end{align}
\setcounter{equation}{\value{MYtempeqncnt}}
\hrulefill
\vspace*{4pt}
\end{figure*}
Equation (\ref{eq:simassocpr}) corroborates the intuition that  increasing association bias $\bias{}$  leads to decrease in the mean  population of macro cell users implied by the decreasing  $\passoc_{1}$. On the other hand, the mean population of range expanded users increases  implied by the increasing $\passoc_{\opueindex}$. 
 Further,  $\passoc_{2}\triangleq \passoc_{\ipueindex} + \passoc_{\opueindex }$ is the  probability of a typical user associating with  the  tier 2.

The conditional $\SINR$ coverage, when a typical user $u\in\set{U}_l$ is $\pcov_l(\SINRthresh)\triangleq \pr\left(\SINR >\SINRthresh| u \in \set{U}_l\right).$
\begin{lem}\label{lem:pcov}
($\SINR$ Coverage) For a typical user in the setup of Sec. \ref{sec:sysmodel}, the $\SINR$ coverage  is
\begin{equation}
\pcov(\SINRthresh)= \passoc_1\pcov_1(\SINRthresh) + \passoc_{\ipueindex}\pcov_{\ipueindex}(\SINRthresh) + \passoc_{\opueindex }\pcov_{\opueindex }(\SINRthresh), 
\end{equation}\addtocounter{equation}{3}
where the conditional $\SINR$ coverage are given by (\ref{eq:pcov1})-(\ref{eq:pcov3}),

 $\Q(a,b,c )= c^{2/b} +  a^{2/b}\int_{(\frac{c}{a})^{2/b}}^\infty \frac{\mathrm{d} u}{ 1+ u^{b/2}}$, 
 and $ \SNR_{k} (y)= \frac{\power{k}\ndistnsc^{-\ple{k}}}{\noisepower}$.
\end{lem}
\begin{IEEEproof}
See Appendix \ref{sec:proofpcov}.
\end{IEEEproof}
The result in Lemma  \ref{lem:pcov} is for the most general case and involves a single numerical integration  along with a lookup table for $\Q$. The expressions can be  further simplified as in the following corollary. 
\begin{cor}\label{cor:pcovequalple}
With noise ignored, $\SNR_k \to \infty$, assuming equal path loss exponents $\ple{k} \equiv \ple{}$, the $\SINR$  coverage  of a typical user is  
\small
\begin{multline}
\pcov(\SINRthresh) = \frac{\dnsty{1}}{\sum_{k=1}^2\dnsty{k}(\power{k}/\power{1})^{2/\ple{}}\Q(\SINRthresh,\ple{},\bias{k})} \\+ \frac{\dnsty{2}}{\sum_{k=1}^2\dnsty{k}(\power{k}/\power{2})^{2/\ple{}}\Q(\SINRthresh,\ple{},1)} \\+ \frac{\dnsty{2}}{ \dnsty{2}\Q(\SINRthresh,\ple{},1)+\dnsty{1}\left\{\power{1}/(\power{2}\bias{2})\right\}^{2/\ple{}}} \\ - \frac{\dnsty{2}}{ \dnsty{2}\Q(\SINRthresh,\ple{},1)+\dnsty{1}(\power{1}/\power{2})^{2/\ple{}}}.
\end{multline}
\normalsize
\end{cor}
As evident from the above Lemma and Corollary, $\SINR$ coverage is independent of the resource partitioning fraction  $\af$  because of the independence of  $\SINR$ on the  amount of resources allocated to a user in our model.  Further, the $\SINR$ distribution of  the small cell users, $\pcov_{\ipueindex}$, is independent of association bias, as $\ipue$ is independent of bias. Further  insights about $\SINR$ coverage are deferred until the next section. In general, we show that $\SINR$ coverage with and without resource partitioning show considerably different behavior, which is also reflected in the rate coverage trends.

\begin{figure*}[!t]
\setcounter{MYtempeqncnt}{\value{equation}}
\setcounter{equation}{28}
\begin{align}
\avpsys_1(\RATEthresh) &= 2\pi\frac{\dnsty{1}}{\passoc_1}\int_{0}^{\infty}\ndistnsc \exp\Bigg\{-\frac{\uRATEthresh(\nRATEthresh\avload{1}\effres{1})}{\SNR_1(\ndistnsc)} -\pi\sum_{k=1}^2\dnsty{k}\npower{k}^{2/\ple{k}}\Q(\uRATEthresh(\nRATEthresh\avload{1}\effres{1}),\ple{k},\nbias{k})\ndistnsc^{2/\nple{k}}\Bigg\}\mathrm{d} \ndistnsc \label{eq:mrcov1}\\
\avpsys_{\ipueindex}(\RATEthresh)&=2\pi\frac{\dnsty{2}}{\passoc_{\ipueindex}}\int_{0}^{\infty}\ndistnsc \exp\Bigg\{-\frac{\uRATEthresh(\nRATEthresh\avload{\ipueindex}\effres{\ipueindex})}{\SNR_2(\ndistnsc)} -\pi\sum_{k=1}^2\dnsty{k}\npower{k}^{2/\ple{k}}\Q(\uRATEthresh(\nRATEthresh\avload{\ipueindex}\effres{\ipueindex}),\ple{k},1)\ndistnsc^{2/\nple{k}}\Bigg\}\mathrm{d} \ndistnsc \label{eq:mrcov2}
\end{align}
\small
\begin{align}
\avpsys_{\opueindex }(\RATEthresh)&=2\pi\frac{\dnsty{2}}{\passoc_{\opueindex }} \int_{0}^{\infty}\ndistnsc \exp\left\{-\frac{\uRATEthresh(\nRATEthresh\avload{\opueindex }\effres{\opueindex })}{\SNR_2(\ndistnsc)} -\pi\dnsty{2}\ndistnsc^2\Q(\uRATEthresh(\nRATEthresh\avload{\opueindex }\effres{\opueindex }),\ple{2},1)-\pi\dnsty{1}\ndistnsc^{2/\nple{1}}\npower{1}^{2/\ple{1}}\right\}\left\{\exp\Big\{-\pi\dnsty{1}\npower{1}^{2/\ple{1}}\ndistnsc^{2/\nple{1}}(\nbias{1}^{2/\ple{1}}-1)\Big\}
-1\right\}\mathrm{d} \ndistnsc \, \label{eq:mrcov3},
\end{align}
\normalsize
\setcounter{equation}{\value{MYtempeqncnt}}
\hrulefill
\vspace*{4pt}
\end{figure*}

\subsection{Main Result}
Similar to the  conditional $\SINR$ coverage, conditional rate coverage, when a typical user $u\in\set{U}_l$ is 
$\rcov_l(\RATEthresh)\triangleq \pr\left(R >\RATEthresh| u \in \set{U}_l\right). $
 The following theorem gives the rate distribution over the entire network. 
\begin{thm}\label{thm:rcov}
(Rate Coverage)  For a typical user in the setup of Sec. \ref{sec:sysmodel}, the  rate coverage is 
\begin{equation}
\psys(\RATEthresh) = \passoc_1\psys_1(\RATEthresh) + \passoc_{\ipueindex}\psys_{\ipueindex}(\RATEthresh) + \passoc_{\opueindex }\psys_{\opueindex }(\RATEthresh) ,
\end{equation}\addtocounter{equation}{3}
where the conditional rate coverage are given by (\ref{eq:rcov1})-(\ref{eq:rcov3}),  $\pmf_{l}(n) \triangleq\pr\left(\load{l}=n\right)$,
  $\uRATEthresh({x}) = 2^{x}-1$, and $\nRATEthresh=\RATEthresh/\res$.
\end{thm}
\begin{IEEEproof}
Using (\ref{eq:sinrdef}) and (\ref{eq:ratedef}), the probability that the rate requirement of a random user $u$  is met is
\small
\begin{align}
\pr(R> \RATEthresh)&= 
\sum_{l \in \{1,\ipueindex,\opueindex\}}\pr(u\in\set{U}_l)\pr\left(\frac{\res}{\effres{l}\load{l}}\log\left(1+\SINR\right)>\RATEthresh|u\in \set{U}_l\right)
\\
& = \sum_{l \in \{1,\ipueindex,\opueindex\}}\passoc_{l} \cexpect{\pcov_{l}\left(\uRATEthresh(\nRATEthresh\load{l}\effres{l})\right)}{\load{l}} ,
\end{align}\normalsize
where $\nRATEthresh = \RATEthresh/\res$ and $\uRATEthresh({x}) = 2^{x}-1$. In general, the load and $\SINR$ are correlated, as APs with larger association regions have higher load and larger user to AP distance (and hence lower $\SINR$). However for tractability of the analysis, this dependence is ignored, as in \cite{SinDhiAnd13}, resulting in $\cexpect{\pcov_{l}(\uRATEthresh(x\load{l}))}{\load{l}}= \sum_{n\ge 1}\pmf_{l}(n)\pcov_{l}\left(\uRATEthresh (x n)\right)$, where $\pmf_{l}(n) =\pr\left(\load{l}=n\right)$.  Using Lemma \ref{lem:pcov}, the rate coverage expression is then obtained.
\end{IEEEproof}
The probability mass function of the load depends on the association area, which needs to be characterized. 
\begin{rem}\label{rem:assoc_area}
(Mean Association Area) Association area of an AP is the area of the corresponding association region.  Using the ergodicity of the PPP, the mean of the association area $\area_k$ of a typical AP of $\uth{k}$ tier is $\expect{\area_k} =\frac{\passoc_{k}}{\dnsty{k}}$.
\end{rem}
The association  region of a tier 2 AP can be further partitioned into two regions. The non-shaded region in Fig. \ref{fig:icic} surrounding a small cell at $x$ can be characterized as
\begin{equation}
\assocr_{x_{\ipueindex}}\triangleq \big\{ y \in \R^2: \|y-x\| \leq \left(\power{2}/\power{1}\right)^{1/\ple{2}}\|y-\NAP{k}(y)\|^{\nple{1}},\,\forall k \big\}.
\end{equation}
As per (\ref{eq:setdef}), all the users lying in $\assocr_{x_{\ipueindex}}$ are the small cell users (belonging to $\ipue$) and  recalling (\ref{eq:assocr}) all users lying in $\assocr_{x_{\opueindex}} \triangleq \assocr_{x_{2}}-\assocr_{x_{\ipueindex}}$ are the offloaded users that belong to   $\opue$. In Fig. \ref{fig:icic},  $\assocr_{x_{\opueindex}}$ is the shaded region  surrounding a tier 2 AP. 

\begin{rem}\label{rem:associatioarea} \textit{(Association Area Distribution)} A linear scaling based approximation for the distribution of association areas proposed in \cite{SinDhiAnd13}, which matched the first moment,  is generalized in this paper  to the setting  of resource partitioning as below
\begin{align}\label{eq:area_approx}
\area_{1} = \area\left(\frac{\dnsty{1}}{\passoc_{1}}\right), \\
\area_{\ipueindex} = \area\left(\frac{\dnsty{2}}{\passoc_{\ipueindex}}\right), \text{ and } \area_{\opueindex}& = \area\left(\frac{\dnsty{2}}{\passoc_{\opueindex}}\right),
\end{align} 
where $\area\left(y\right)$ is the area of a typical cell of a Poisson Voronoi (PV) of density $y$ (a scale parameter). 
\end{rem}

Using the area distribution proposed in \cite{Ferenc2007} for  PV $\area(y)$, the following lemma characterizes the probability mass function (PMF) of  the load seen by a typical user.
\begin{lem}\label{lem:oloadpgf}
(Load PMF) The PMF  of the load at tagged AP of a typical user $u\in \set{U}_l$  is
\begin{multline}\label{eq:pgfload}
\pmf_{l}(n) \triangleq\pr\left(\load{l}=n\right)\\=\frac{3.5^{3.5}}{(n-1)!}\frac{\Gamma(n+3.5)}{\Gamma(3.5)}\left(\frac{\userdnsty\passoc_{l}}{\dnsty{\tmap{l}}}\right)^{n-1}\left(3.5 + \frac{\userdnsty\passoc_{l}}{\dnsty{\tmap{l}}}\right)^{-(n+3.5)}  \\n\geq 1,
\end{multline} 
where $\Gamma(x)=\int_{0}^\infty \exp(-t)t^{x-1}\mathrm{d}t$ is the gamma function.
\end{lem}
\begin{IEEEproof}
See Appendix \ref{sec:proofoloadpgf}.
\end{IEEEproof}

The rate distribution expression for the most general setting requires a single numerical integral after use of lookup tables for $\Q$ and $\Gamma$.  The summation over $n$ in Theorem \ref{thm:rcov} can be accurately approximated as a finite summation  to a sufficiently large value, ${n}_\text{max}$ (say), since both the terms $\pmf_{l}(n)$ and $\pcov_{l}\left(\uRATEthresh(xn))\right)$ decay rapidly for large $n$.   

The rate coverage expression can be further simplified if the load at each AP is assumed to equal its mean. 
\begin{cor}\label{cor:rcovmeanload}
(Mean Load Approximation) Rate coverage with the mean load approximation is given by 
\begin{equation}\label{eq:rcovmeanload}
\avpsys(\RATEthresh) = \passoc_1\avpsys_1(\RATEthresh) + \passoc_{\ipueindex}\avpsys_{\ipueindex}(\RATEthresh) + \passoc_{\opueindex }\avpsys_{\opueindex }(\RATEthresh) ,
\end{equation}
where 
the conditional rate coverage are given by (\ref{eq:mrcov1})-(\ref{eq:mrcov3}) \addtocounter{equation}{3}
 and $\avload{l}= \expect{\load{l}} = 1+ \frac{1.28\userdnsty\passoc_{l}}{\dnsty{\tmap{l}}}.$
\end{cor}
\begin{IEEEproof}
Lemma \ref{lem:oloadpgf} gives the first moment of load as
$\expect{\load{l}} =  1+ \frac{\userdnsty\passoc_{l}}{\dnsty{\tmap{l}}}\expect{\area^2(1)}$.
Further, using the result that $\expect{\area^2(1)}=1.28$ \cite{Gil62}, 
along with  an approximation 
$\cexpect{\pcov_{k}\left(\uRATEthresh(x\load{k})\right)}{\load{k}} \approx
\pcov_{k}\left(\uRATEthresh(x\expect{\load{k}})\right)$, the simplified rate coverage expression is obtained.
\end{IEEEproof}
The mean load approximation above simplifies the rate coverage expression by eliminating the summation over $n$. The numerical integral can also be eliminated  by ignoring noise and assuming equal path loss exponents (as is done in Sec \ref{sec:rcovtrend}). 
As  can be observed from Theorem \ref{thm:rcov} and Corollary \ref{cor:rcovmeanload}, the rate coverage for range expanded users $\rcov_{\opueindex }$ increases with increase in resource partitioning fraction $\af$, as  users in $\opue$ can be scheduled on a larger fraction of (macro) interference free resources. On the other hand, the rate coverage for the macro users $\rcov_{1}$ and small cell (non-range expanded) users $\rcov_{\ipueindex}$ decreases with the corresponding increase. Further insights on the effect of biasing are delegated to the next section. 
\begin{figure*}
\setcounter{MYtempeqncnt}{\value{equation}}
\setcounter{equation}{37}
\begin{align}
\pcov_{\ipueindex j}(\SINRthresh)& =2\pi\frac{\dnsty{2}}{\passoc_{\ipueindex j}}\int_{0}^{\infty}\ndistnsc \exp\left\{-\frac{\SINRthresh}{\SNR_j(y)} -\pi\left(\sum_{k\neq j}\dnsty{k}\npower{k}^{2/\ple{k}}\Q(\SINRthresh,\ple{k},\bias{k})\ndistnsc^{2/\nple{k}}+ \dnsty{j}\Q(\SINRthresh,\ple{j},1)\ndistnsc^{2} \right)\right\}\mathrm{d} \ndistnsc \label{eq:pcovj1}\\
\pcov_{\opueindex j }(\SINRthresh) & =2\pi\frac{\dnsty{2}}{\passoc_{\opueindex j}}\int_{0}^{\infty}\ndistnsc \exp\left\{-\frac{\SINRthresh}{\SNR_j(y)}-\pi\left(\sum_{k\geq2}\dnsty{k}\npower{k}^{2/\ple{k}}\Q(\SINRthresh,\ple{k},\nbias{k})\ndistnsc^{2/\nple{k}}-\dnsty{1}(\npower{1}\nbias{1})^{2/\ple{1}}y^{2/\nple{1}}\right)\right\}\nonumber \\
&\times\prod_{k\neq j}\left\{1- \exp\left(-\pi\dnsty{k}(\npower{k}\nbias{k})^{2/\ple{k}}\ndistnsc^{2/\nple{k}}(\bias{j}^{2/\ple{k}}-1)\right)
\right\}\mathrm{d} \ndistnsc \label{eq:pcovj2}\,,
\end{align}
\normalsize
\setcounter{equation}{\value{MYtempeqncnt}}
\hrulefill
\vspace*{4pt}
\end{figure*}
\subsection{Rate Coverage with Limited Backhaul Capacities}\label{sec:bkhl}
Analysis in the previous sections assumed infinite backhaul capacities and thus the air interface was the only bottleneck affecting downlink rate. 
However, with limited backhaul capacities $\bkhl{k}$ for BSs of tier $k$, the rate is given by 
\begin{equation}\label{eq:ratebkhl}
R' = 
\indic(u\in\mue)\min\left(R,\frac{\bkhl{1}}{\load{1}}\right) + \indic(u\in\pue)\min\left(R,\frac{\bkhl{2}}{\load{2}}\right),
\end{equation}
where $R$ is the rate of the user with infinite backhaul bandwidth. 
The above rate allocation assumes that the available backhaul bandwidth for a BS of tier $k$, $\bkhl{k}$, is shared equally among the associated users/load $\load{k}$. This allocation model is similar to the fair round robin scheduling and results in the peak rate of a typical user (associated with an AP of tier $k$) being capped at $\frac{\bkhl{k}}{\load{k}}$. The analysis can be extended to incorporate a generic peak rate dependency $f(\bkhl{k},\load{k})$ on backhaul bandwidth and load at the AP (which may result from a different backhaul allocation strategy)\footnote{Exact analysis of wired backhaul allocation among the competing TCP flows could be an area of future investigation.}. 
The following lemma gives the rate distribution in this setting.
\begin{lem}\label{lem:ratebkhl}
(Rate Coverage with Limited Backhaul)
The rate  coverage in the setting of Sec. \ref{sec:sysmodel} and with rate model of (\ref{eq:ratebkhl}) is
\begin{equation}
\psys^{'}(\RATEthresh) = \pr(R'> \RATEthresh)=\passoc_1\psys_1^{'}(\RATEthresh) + \passoc_{\ipueindex}\psys_{\ipueindex}^{'}(\RATEthresh) + \passoc_{\opueindex }\psys_{\opueindex }^{'}(\RATEthresh) ,
\end{equation}
where 
\begin{equation}
\begin{aligned}
\psys_{1}^{'}(\RATEthresh)&= \sum_{n=1}^{\lceil\bkhl{1}/\RATEthresh-1\rceil}\pmf_1(n)\pcov_1\left(\uRATEthresh(\effres{1} n\nRATEthresh)\right),\\
\rcov_{\ipueindex}^{'}(\RATEthresh)&= \sum_{m=0}^{\lceil\bkhl{2}/\RATEthresh-2\rceil}\pmf_{\opueindex}(m)\sum_{n=1}^{\lceil\bkhl{2}/\RATEthresh-m-1\rceil}\pmf_{\ipueindex}(n)\pcov_{\ipueindex}\left(\uRATEthresh(\effres{\ipueindex} n\nRATEthresh)\right),  \\ \rcov_{\ipueindex}^{'}(\RATEthresh)&= \sum_{m=0}^{\lceil\bkhl{2}/\RATEthresh-2\rceil}\pmf_{\ipueindex}(m)\sum_{n=1}^{\lceil\bkhl{2}/\RATEthresh-m-1\rceil}\pmf_{\opueindex}(n)\pcov_{\opueindex}\left(\uRATEthresh(\effres{\opueindex} n\nRATEthresh)\right),
\end{aligned}
\end{equation}
and  $\pcov_l$ is given by Lemma \ref{lem:pcov}.
 
\end{lem}
\begin{IEEEproof}
Since the maximum rate of a user $u\in \set{U}_l$ is $\bkhl{\tmap{l}}/\load{\tmap{l}}$. Thus, for this  user to have positive rate coverage, i.e.,  $\pr(R> \RATEthresh) >0  $, a necessary condition is $\load{\tmap{l}}\leq \lceil\frac{\bkhl{\tmap{l}}}{\RATEthresh}-1\rceil$. When this necessary condition is satisfied, the rate coverage is equivalent to
\begin{align}
\rcov_{l}{'}&= \pr\left(R^{'}> \RATEthresh\right)\nonumber\\&= \pr\left(\frac{\res}{\effres{l}\load{l}}\log(1+\SINR)>\RATEthresh\cap\left\{\load{\tmap{l}}< \frac{\bkhl{\tmap{l}}}{\RATEthresh}\right\}\right)
\end{align}

Using $\load{l}=\load{\tmap{l}}$ for $l=1$, and independence of $\load{\ipueindex}$ and $\load{\opueindex}$ the conditional rate coverage are obtained.
\end{IEEEproof}
It is evident from the above Lemma that rate coverage decreases with decreasing backhaul bandwidth. Therefore, decreasing $\bkhl{2}$ will lead to decrease in the rate of the user when it is associated to small cell and thus decreasing the optimal offloading bias (this is further explored in subsequent sections).  As the backhaul bandwidth increases to infinity, Lemma \ref{lem:ratebkhl} leads to Theorem \ref{thm:rcov}, or, $\lim_{\bkhl{\tmap{l}}\to \infty} \psys_l^{'}\to\psys_l$.

\subsection{Extension to Multi-tier Downlink}\label{sec:multier}
The analysis in the previous sections discussed a two-tier setup, which can be  generalized to  a $K$-tier ($K >2$) setting. In this setting, location of the BSs of $\uth{k}$ tier are assumed according to a PPP $\tierPPP{k}$ of density $\dnsty{k}$. Further, $\bias{k}$ is assumed to be the  association bias corresponding to tier $k$, where $\bias{1} = 0$ dB and $\bias{k} \geq 0$ dB $\forall k>1$. Similar to (\ref{eq:setdef}), a user $u$ associated with tier $j$ can be classified into two disjoint sets:
\begin{equation}\label{eq:setdefk}
u \in \begin{cases}
 \set{U}_{\bar{B}j} \text{ if    $\power{j}\NDIST{j}^{-\ple{j}}>\power{k}\bias{k}\NDIST{k}^{-\ple{k}}$ } \forall k\neq j \\
 \set{U}_{{B}j} \text{  if $u \notin \set{U}_{\bar{B}j}$ and $\power{j}\bias{j}\NDIST{j}^{-\ple{j}}>\power{k}\bias{k}\NDIST{k}^{-\ple{k}}$ } \forall k\neq j.
 \end{cases}
\end{equation}
 With resource partitioning, an AP of tier $j$ schedules the offloaded users,  $\set{U}_{{B}j}$,  in $\af$ fraction of the resources, which are protected from the macro-tier interference and the non-range expanded users are scheduled on $1-\af$ fraction of the resources. Thus, the $\SINR$ of a user $u$ associated with tier $j$ is 

\small
\begin{equation}\label{eq:sinrdefk}
\SINR = \indic\left(u \in \set{U}_{{B}j}\right)\frac{\power{j}\chanl_{\ndistnsc} \ndistnsc^{-\ple{j}}}{\sum_{k=2}^K I_{y,k} + \noisepower}+ \indic\left(u \in \set{U}_{\bar{B}j}\right)\frac{\power{j}\chanl_{\ndistnsc} \ndistnsc^{-\ple{j}}}{\sum_{k=1}^K I_{y,k} + \noisepower}.
\end{equation}
\normalsize
\addtocounter{equation}{2}
By using similar techniques as in a two-tier setting, the $\SINR$ coverage for this setting is given in (\ref{eq:pcovj1})-(\ref{eq:pcovj2}). The rate is given  by
\small
\begin{equation}\label{eq:ratedefk}
\rate{} = \left\{\indic\left(u \in \set{U}_{{B}j}\right)\frac{\af}{\load{Bj}}
+ \indic\left(u \in \set{U}_{\bar{B}j}\right)\frac{1-\af}{\load{\bar{B}j}}\right\}
\res\log\left(1+\SINR\right). 
\end{equation}\normalsize 
The rate coverage for this setting  can be derived by using (\ref{eq:pcovj1})-(\ref{eq:pcovj2}) and a generalization of Lemma \ref{lem:oloadpgf}.

\subsection{Validation of Analysis}\label{sec:validation}
We verify the developed analysis, in particular Theorem \ref{thm:rcov}, Corollary \ref{cor:rcovmeanload}, and Lemma \ref{lem:ratebkhl}, in this section. The rate distribution  is validated by  sweeping over a range of rate thresholds.    The rate distribution obtained through simulation and that from Theorem \ref{thm:rcov} and Corollary \ref{cor:rcovmeanload} for two  values for the pair of  bias and resource partitioning fraction ($\bias{},\af$) is shown in Fig.~\ref{fig:ratedist_twotier_icic}. The respective densities used are $\dnsty{1} = 1$ BS/km$^2$, $\dnsty{2} = 5$ BS/km$^2$, and $\userdnsty = 100 $ users/km$^2$ with $\ple{1} =3.5$, $\ple{2} =4$. The assumed transmit powers are $\power{1} =46$ dBm and $\power{2}=26$ dBm. Thermal noise power is assumed to be $\noisepower=-10$ dBm.
The rate distribution for the case with limited backhaul  obtained through simulation and that from Lemma  \ref{lem:ratebkhl} is shown in Fig. \ref{fig:ratedist_twotierlimitedbkhl}. The rate distribution is shown for two different backhaul bandwidths for a bias  of $\bias{}=10$ dB and without resource partitioning. 
Both the plots show that the analytical results, Theorem \ref{thm:rcov} and Lemma \ref{lem:ratebkhl}, give quite accurate (close to simulation) rate distribution. Furthermore, the mean  load approximation based Corollary \ref{cor:rcovmeanload} is also not that far off from the exact curves in Fig. \ref{fig:ratedist_twotier_icic}. This gives further confidence  that the rate distribution obtained with mean load approximation in Corollary \ref{cor:rcovmeanload} can be used for further insights (as is done in the following sections).

\begin{figure*}
  \centering
\subfloat[]
{\label{fig:ratedist_twotier_icic}\includegraphics[width=\columnwidth]{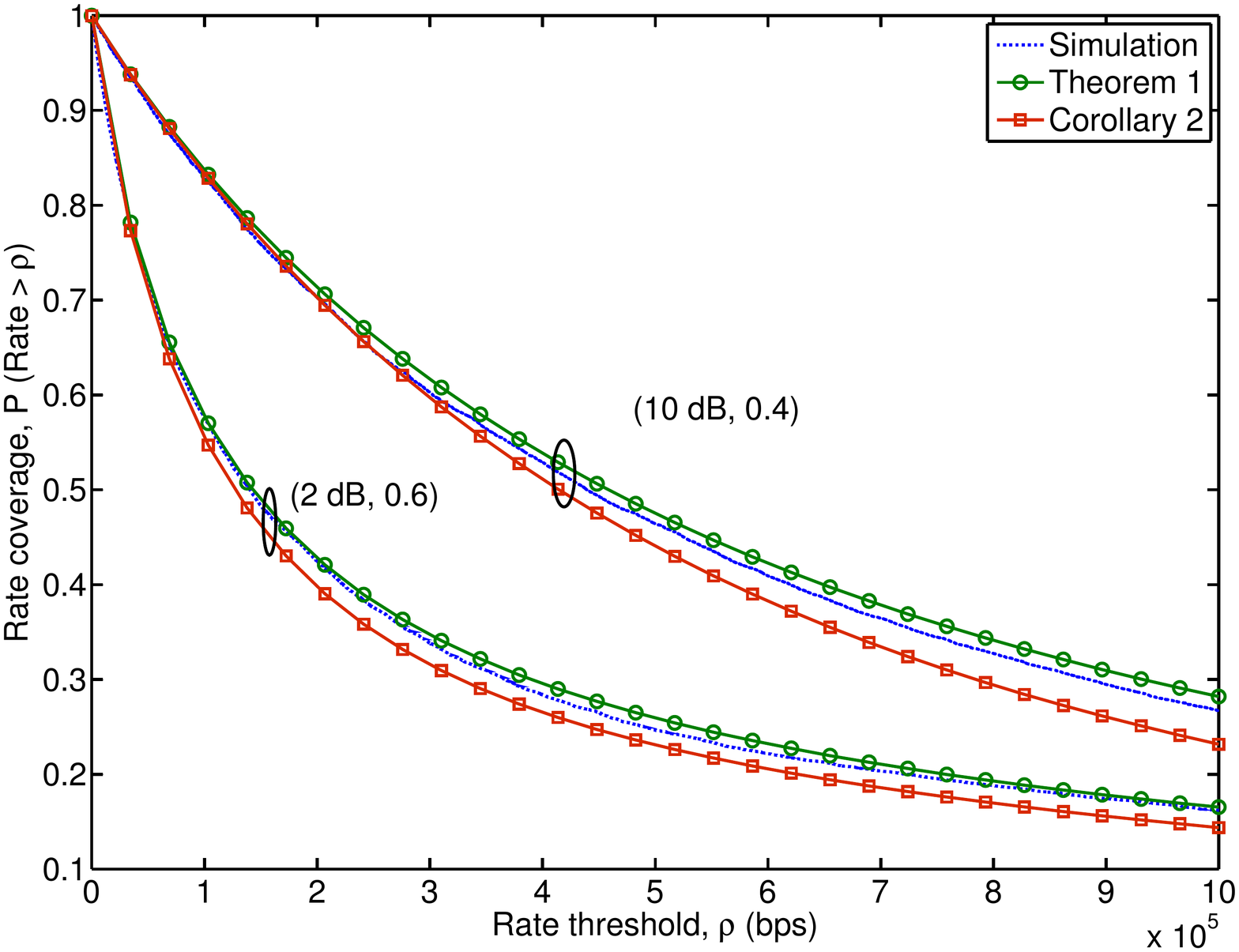}}
\subfloat[]{\label{fig:ratedist_twotierlimitedbkhl}\includegraphics[width=\columnwidth]{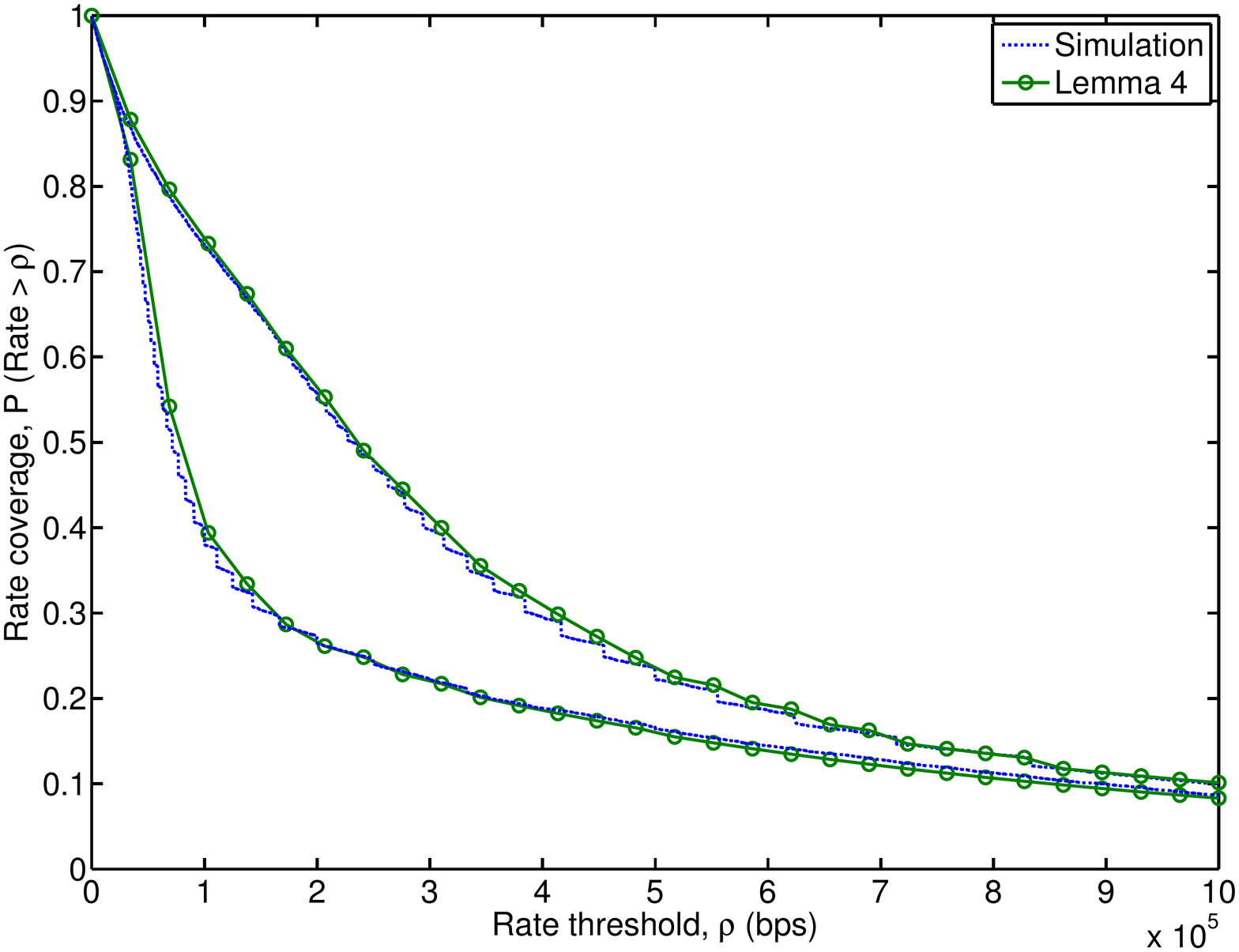}}
\caption{(a) Rate distribution obtained from simulation, Theorem \ref{thm:rcov} and Corollary \ref{cor:rcovmeanload} for $\dnsty{2}=5\dnsty{1}$, $\ple{1}=3.5$, and $\ple{2}=4$. (b) Rate distribution obtained from simulation and Lemma 4 for $\dnsty{2}=5\dnsty{1}$, $\ple{1}=3.5$, and $\ple{2}=4$.}
 \label{fig:validation}
\end{figure*}

\section{Insights on optimal $\SINR$ and Rate Coverage}\label{sec:optbias}
As it was mentioned earlier  the extent of resource partitioning and offloading needs to  be carefully chosen for optimal performance. 
Although a simplified setting is considered in the following results for analytical  insights, it is shown that these insights extend to more general settings  through numerical results. 

\subsection{$\SINR$  Coverage: Trends and Discussion}\label{sec:pcovtrend}
Although rate coverage is the main metric of interest, insights obtained from $\SINR$ coverage should be useful in explaining key trends in rate coverage. As stated before, the $\SINR$ coverage with and without resource partitioning exhibits different behavior in conjunction with offloading. The following lemma presents some key trends for $\SINR$ coverage in both settings.

\begin{cor}\label{lem:pcovcomp}
Ignoring thermal noise ($\noisepower\to 0$),  assuming equal path loss exponents and equal to four  ($\ple{k}\equiv 4$ \footnote{$\alpha$ typically  varies from $3$ to $4$ depending on the propagation environment.}), the $\SINR$  coverage without resource partitioning is  
\small
\begin{align}\label{eq:pcovwsimple}
&\pcov^w(\SINRthresh) = \frac{1}{\sqrt{\SINRthresh}\tan^{-1}(\sqrt{\SINRthresh})+1 + a\sqrt{p}(\sqrt{\SINRthresh}\tan^{-1}(\sqrt{\SINRthresh/b})+ \sqrt{b})}\nonumber\\
& + \frac{1}{\sqrt{\SINRthresh}\tan^{-1}(\sqrt{\SINRthresh})+1 + \frac{1}{a\sqrt{p}}(\sqrt{\SINRthresh}\tan^{-1}(\sqrt{b\SINRthresh})+ \sqrt{1/b})},
\end{align}
\normalsize
where $b=\frac{\bias{2}}{\bias{1}}$, $a=\frac{\dnsty{2}}{\dnsty{1}}$, and  $p=\frac{\power{2}}{\power{1}}$. The $\SINR$ coverage with resource partitioning  for the corresponding setting is 
\small
\begin{align}\label{eq:pcovsimple}
&\pcov(\SINRthresh)=  \frac{1}{\sqrt{\SINRthresh}\tan^{-1}(\sqrt{\SINRthresh})+1 + a\sqrt{p}(\sqrt{\SINRthresh}\tan^{-1}(\sqrt{\SINRthresh/b})+ \sqrt{b})}\nonumber\\& +
 \frac{1}{\sqrt{\SINRthresh}\tan^{-1}(\sqrt{\SINRthresh})+1 + \frac{1}{a\sqrt{p}}(\sqrt{\SINRthresh}\tan^{-1}(\sqrt{\SINRthresh})+ 1)} 
 \nonumber\\&  + \frac{1}{\sqrt{\SINRthresh}\tan^{-1}(\sqrt{\SINRthresh})+1 + \frac{1}{a\sqrt{pb}}} - \frac{1}{\sqrt{\SINRthresh}\tan^{-1}(\sqrt{\SINRthresh})+1 + \frac{1}{a\sqrt{p}}} .
\end{align}\normalsize
\end{cor}
\begin{proof}
Using \[\Q(\SINRthresh,4,x)= 1 + \sqrt{\SINRthresh}\int_{\sqrt{\frac{x}{\SINRthresh}}}^\infty\frac{\mathrm{d} u}{1+u^2}= 1+ \sqrt{\SINRthresh}\tan^{-1}(\sqrt{\SINRthresh/x}),\]
in Corollary \ref{cor:pcovequalple},  and substituting $a=\frac{\dnsty{2}}{\dnsty{1}}$,   $p=\frac{\power{2}}{\power{1}}$, and $b=\frac{\bias{2}}{\bias{1}}$, the expression in  (\ref{eq:pcovsimple}) is obtained. For the case with no resource partitioning, the expression derived in Lemma 5 of  \cite{SinDhiAnd13}  can be simplified  using similar techniques  to give (\ref{eq:pcovwsimple}). 
\end{proof}
Moreover,  in this setting the following  three claims can be made:
\begin{itemize}
\item[] \textbf{Claim 1}:
 Offloading with a bias ($b>1$) leads to  suboptimal $\SINR$ coverage $\pcov^w$ in the case of no resource partitioning and $\SINRthresh\geq1$ ($0$ dB). 
\item[] \textbf{Claim 2}:
 With resource partitioning,  the bias $b$ maximizing  the $\SINR$ coverage $\pcov$ can be greater than $0$ dB, the upper bound on which,  however, decreases with increasing density of small cells.
\item[] \textbf{Claim 3}: With resource partitioning, the $\SINR$ coverage obtained by offloading all the users to small cells, i.e., $b \to \infty$,  is always less than that of no biasing, i.e., $b=1$.
\end{itemize}
\begin{proof}
See Appendix \ref{sec:proof_pcovcomp}.
\end{proof}

From the above corollary, it can  be noted that  $\pcov|_{b=1}= \pcov^w|_{b=1}$, i.e., with no biasing/offloading $\SINR$ distribution with and without resource partitioning are  equal  as the orthogonal resource is not utilized by any user. Further, with resource partitioning  the contribution to coverage from range expanded users (third term in  (\ref{eq:pcovsimple})) increases with increasing  bias, whereas the  corresponding contribution from the macro cell users (first term in (\ref{eq:pcovsimple})) decreases with  increasing bias. This is a bit counter-intuitive as one would expect with increasing bias only very good geometry users remain with macro cell. The  reasoning behind this is  the \textit{ fast}  decrease in the fraction of such users $\passoc_{1}$ with increasing bias, which leads to an overall decrease  in the coverage contribution from macro cell users.  Similarly, the corresponding \textit{fast} increase in  the fraction of offloaded users $\passoc_{\opueindex }$ leads to an overall increase in the their contribution to coverage. A similar trend is observed for coverage without resource partitioning  in (\ref{eq:pcovwsimple}).

An intuitive explanation for the  claims is as follows. Claim 1  states that without resource partitioning, proactively offloading a user to small cell through an association bias is suboptimal, as the user would then always be  associated  to an AP offering lower $\SINR$.  On the other hand, with resource partitioning certain fraction of  users can be offloaded to  small cells and served on the resources which are protected  from macro cell interference. In this case,  increasing the small cell density, however, increases the interference on the orthogonal resources for the offloaded users, and hence, the optimal offloading bias is forced downward. Claim 2 justifies the described intuition when the bound is tight.  Claim 3  suggests that preventing only offloaded users from macro tier interference is clearly suboptimal when almost all the macro cell  users are offloaded to small cells.  Of course, with large association bias, $b\to \infty$, it would be better to shut the macro tier completely off, i.e., $\af=1$.

The above discussion is corroborated by the results in Fig. \ref{fig:pcov_trend_icic}, which shows the effect of association bias on $\SINR$ coverage with varying density of small cells  for a setting with $\ple{1}=3.5$, $\ple{2} =4$, and $\SINRthresh = 0.5$ ($-3$ dB).  Without resource partitioning,  it can be seen that any bias is suboptimal whereas  for the case with resource partitioning  optimal $\SINR$ coverage decreases with increasing density and  the optimal bias also decreases. In case of no resource partitioning, increase in coverage with density is observed due to higher path loss exponents of small cells  (this was also observed in \cite{josanxiaand12}). Thus, as evident from these results,  the above claims hold in general settings too.
\begin{figure}
	\centering
		\includegraphics[width=\columnwidth]{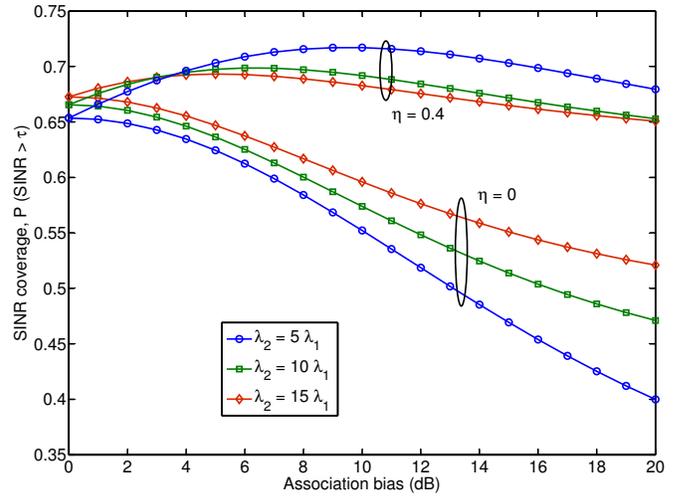}
		\caption{Effect of small cell density on SINR coverage, with and without resource partitioning, as association bias is varied.}
	\label{fig:pcov_trend_icic}
	\end{figure}
%
\subsection{Rate Coverage: Trends and Discussion}\label{sec:rcovtrend}
The following Corollary provides the rate coverage expressions for a simplified setting,   which  is used for drawing the following insights. 
\begin{cor}\label{cor:rcovcomp}
Ignoring thermal noise ($\noisepower\to 0$),  assuming equal path loss exponents and equal to four  ($\ple{k}\equiv 4$),  the rate coverage without resource partition and with the mean load approximation  is  \small
\begin{align}\label{eq:rcovwsimple}
\bar{\rcov}^w &= \frac{1}{\sqrt{u_1}\tan^{-1}(\sqrt{u_1})+1 + a\sqrt{p}(\sqrt{u_1}\tan^{-1}(\sqrt{u_1/b})+ \sqrt{b})} \nonumber \\
&+ \frac{1}{\sqrt{u_2}\tan^{-1}(\sqrt{u_2})+1 + \frac{1}{a\sqrt{p}}(\tan^{-1}(\sqrt{b u_2})+ \sqrt{1/b})}.
\end{align}\normalsize
where  $u_k = \uRATEthresh(\nRATEthresh\avload{k})$, $b=\frac{\bias{2}}{\bias{1}}$, $a=\frac{\dnsty{2}}{\dnsty{1}}$, and  $p=\frac{\power{2}}{\power{1}}$.
The rate coverage with resource partitioning  under the  corresponding assumptions  is 
\small
\begin{align}\label{eq:rcovsimple}
&\bar{\rcov}=  \frac{1}{\sqrt{u_1}\tan^{-1}(\sqrt{u_1})+1 + a\sqrt{p}(\sqrt{u_1}\tan^{-1}(\sqrt{u_1/b})+ \sqrt{b})} \nonumber\\&+ 
 \frac{1}{\sqrt{u_{\ipueindex}}\tan^{-1}(\sqrt{u_{\ipueindex}})+1+ \frac{1}{a\sqrt{p}}(\tan^{-1}(\sqrt{u_{\ipueindex}})+ 1)} 
 \nonumber\\&  + \frac{1}{\sqrt{u_{\opueindex }}\tan^{-1}(\sqrt{u_{\opueindex }})+1 + \frac{1}{a\sqrt{pb}}} \nonumber\\& - \frac{1}{\sqrt{u_{\opueindex }}\tan^{-1}(\sqrt{u_{\opueindex }})+1 + \frac{1}{a\sqrt{p}}}, 
\end{align}
\normalsize
where $u_l =\uRATEthresh(\nRATEthresh\avload{l}\effres{l})$.
\end{cor}
\begin{proof}
For the case with resource partitioning, the rate coverage expression   follows from Corollary \ref{cor:rcovmeanload}   using similar techniques as in the proof of Corollary \ref{lem:pcovcomp}. Without any resource partitioning Corollary 2 of \cite{SinDhiAnd13} is used. 
\end{proof}
From the above expressions it can be observed   that $\bar{\rcov}|_{\af=0, b=1}= \bar{\rcov}^w|_{b=1}$, i.e.,  the  rate distribution  for both scenarios is same  when  no orthogonal resource is made available and there are no offloaded users. 
For the case with no resource partitioning, the contribution to rate coverage from macro cell users (first term of  (\ref{eq:rcovwsimple})) increases initially with increasing bias  as the number of users sharing  the  radio resources at each macro BS decrease. But, beyond a certain  association bias, due to the decreasing fraction of macro cell users, the overall contribution of the corresponding term  towards rate coverage decreases.  Similar trend is shown by the contribution from  small cell users (second term in (\ref{eq:rcovwsimple})). The initial increase with bias is due to the increasing fraction of small cell users and the subsequent decrease is due to increased number of users sharing the  radio resources.  This behavior of rate coverage could be seen as an intuitive reasoning behind the  existence of an optimal bias. 

With resource partitioning,  decreasing  $\af$ increases the rate coverage of macro cell users and small cell users (first two terms of (\ref{eq:rcovsimple}) respectively),  whereas that of range expanded users decreases (last two terms of (\ref{eq:rcovsimple})), due to the decrease in available radio resources. With increasing bias, the rate coverage contribution from  small cell users remains invariant (second term in (\ref{eq:rcovsimple})), as the set $\ipue$ is independent of association bias.  The contribution to rate coverage  from the macro cell users (first term in (\ref{eq:rcovsimple})) and that  from offloaded users (sum of third and fourth term in (\ref{eq:rcovsimple}))  show similar variation with association bias as in the case of  no resource partitioning. Therefore, there should exist an optimal bias for each $\af$ in this setting too. 

The discussion in the above paragraphs is extended further in the following sections where the impact of various  factors is studied on optimal offloading.  For the following results, the parameters   used are the same as in Section \ref{sec:validation} with rate threshold  $\RATEthresh=250$ Kbps wherever applicable.

\subsubsection{Impact of resource partitioning}
 The effect of association bias and resource partitioning  fraction on rate coverage, as obtained from Theorem \ref{thm:rcov}, is shown in Fig. \ref{fig:rcov_trend_bias}. The optimal  pair for this setting  is  $\bias{}= 15$ dB and $\af=0.47$ (obtained by two-fold search), which  gives a  significant increase in the rate coverage compared to the case with  no resource partitioning and offloading ($\bias{}=0$ dB, $\af=0$). With increasing resource partitioning  fraction, the optimal association bias increases as more resources (macro interference free) become available for offloaded users. The variation of rate coverage with resource partitioning fraction for different association  biases is  shown in Fig. \ref{fig:rcov_trend_activity}. The optimal resource partitioning fraction increases with increase in association bias as more resources are needed to serve the increasing number of offloaded users.  As shown, at lower association bias lower resource partitioning fraction  is better as there are not \textit{enough} range expanded users to take advantage of  the resources obtained from  muting the macro tier.  

 A  trend similar to rate coverage can  also be seen in the $5^\mathrm{th}$ percentile rate $\RATEthresh_{95}$ (where $\psys(\RATEthresh_{95}) =0.95$, i.e., fifth percentile of the population receives rate less than 
$\RATEthresh_{95}$) in Fig. \ref{fig:fiveprcntilerate_twotier_icic}.  The corresponding effect on median rate is shown in Fig. \ref{fig:fiftyprcntilerate_twotier_icic}. The optimal  pair of ($\bias{}$,$\af$) for these two metrics is same as that in  rate coverage result. This shows that a single choice of the operating region provides a network-wide optimal performance across different metrics. 
	\begin{figure}
	\centering
		\includegraphics[width=\columnwidth]{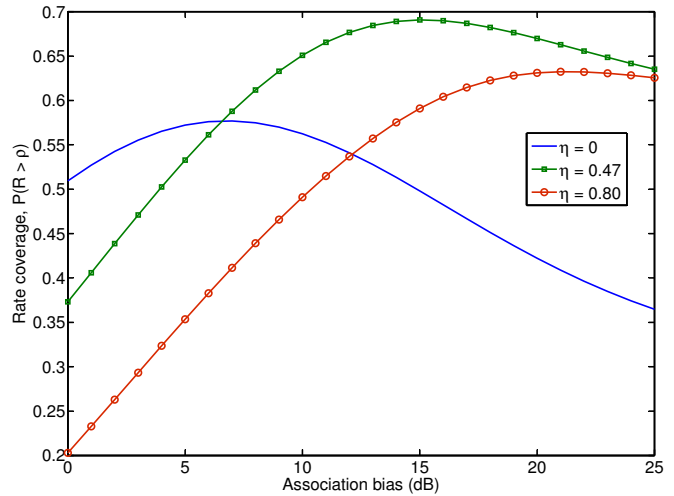}
		\caption{Effect of  association bias,  $\bias{}$, on rate coverage with $\dnsty{2}=5\dnsty{1}$.}
	\label{fig:rcov_trend_bias}
\end{figure}

\begin{figure}
	\centering
		\includegraphics[width=\columnwidth]{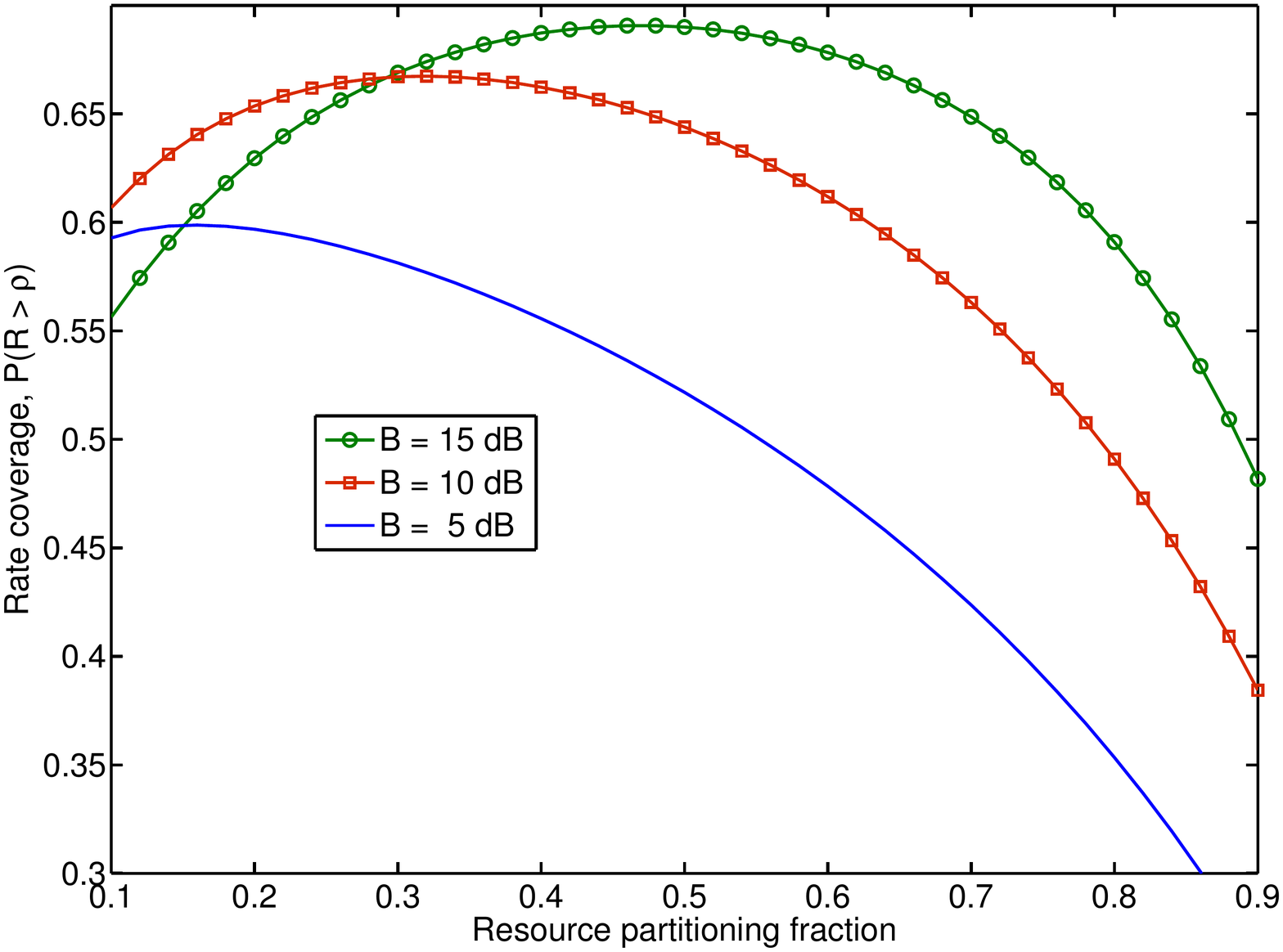}
		\caption{Effect of resource partitioning  fraction, $\af$, on rate coverage with $\dnsty{2}=5\dnsty{1}$.}
	\label{fig:rcov_trend_activity}
\end{figure}

\subsubsection{Impact of infrastructure density}
The impact of density of small cells on the fifth percentile rate is shown in Fig. \ref{fig:fiveprcntilerate_trend_density}. It can be observed that at any particular association bias, as small cell density increases, $\RATEthresh_{95}$ also increases because of the decrease in  load  at  each  AP.  With no resource partitioning, $\eta=0$, the  optimal  bias is seen to be invariant (at $5$ dB) to  the small cell density. 
Similar trend was also  observed in \cite{ye2012user} through exhaustive simulations. However, with  resource partitioning, $\eta > 0$,  optimal association bias decreases with increasing small cell density. The optimal resource partitioning fractions (also shown for each density value) decrease with increasing small cell density.  These observations regarding the behavior of bias and resource partitioning fraction with small cell density can be explained by re-highlighting  the learning from Sec. \ref{sec:pcovtrend} about optimal bias for $\SINR$ coverage. Without any resource partitioning, the optimal bias for $\SINR$ coverage is  $0$ dB and  independent of small cell density and similar independence  is seen for rate coverage where the optimal bias is $5$ dB. The insight of strictly suboptimal performance by a positive bias from $\SINR$, though is  clearly not valid for rate.   With resource partitioning,  increasing small cell density decreased the $\SINR$ coverage  due to the increased interference in the orthogonal time/frequency resources allocated to range expanded users. Similar decrease of   optimal association bias with increasing small cell density is seen for rate for same reasons. The increased interference in the orthogonal resources also leads to the decrease in the optimal fraction of such resources, $\af$. 

 It  is worth pointing here that, although, the optimal association bias decreases with increasing small cell density, the optimal traffic offload fraction $\passoc_2$ increases. With increasing density, at the same traffic offload fraction, the load per AP decreases, increasing the affinity of users for the corresponding tier.   This trend is shown in Fig. \ref{fig:bias_trafficfraction_trend} for two cases -- (1) infinite backhaul bandwidth and (2) limited small cell backhaul bandwidth  $\bkhl{2}=5 $ Mbps. Decreasing  backhaul bandwidth lowers both the optimal bias and consequently optimal offloading fraction at a fixed density.  Interestingly, with increasing density, the optimal bias for case (2) increases, in contrast to the trend in case (1). This is because of the decreasing load per AP and the, thus, diminishing performance gap between the limited and infinite backhaul setting. Therefore, beyond a certain density the gap almost  vanishes  and  the corresponding optimal biases reduce with increasing density.

\begin{figure}
	\centering
		\includegraphics[width=\columnwidth]{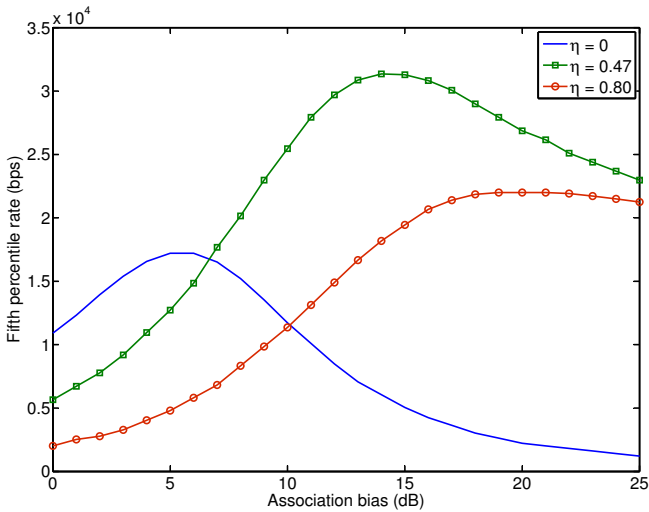}
		\caption{Effect of association bias and resource partitioning fraction $(\bias{}, \af)$ on fifth percentile rate.}
	\label{fig:fiveprcntilerate_twotier_icic}
\end{figure}

%
	\begin{figure}
	\centering
		\includegraphics[width=\columnwidth]{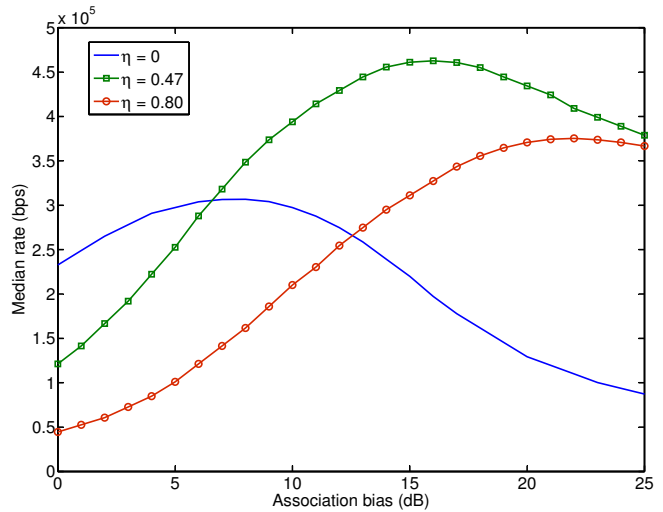}
		\caption{Effect of association bias and resource partitioning  fraction  $( \bias{}, \af)$ on median rate.}
	\label{fig:fiftyprcntilerate_twotier_icic}
\end{figure}


	\begin{figure}
	\centering
		\includegraphics[width=\columnwidth]{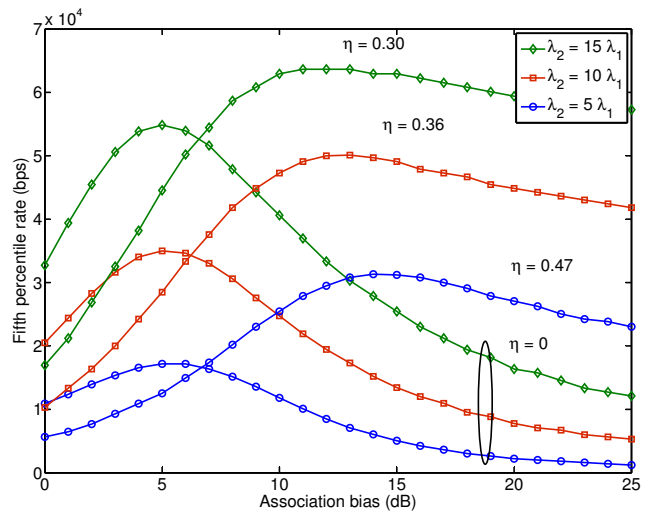}
		\caption{Variation in fifth percentile rate with association bias and resource partitioning  fraction $(\bias{}, \af)$ for different small cell densities.}
	\label{fig:fiveprcntilerate_trend_density}
\end{figure}

\begin{figure}
	\centering
		\includegraphics[width=\columnwidth]{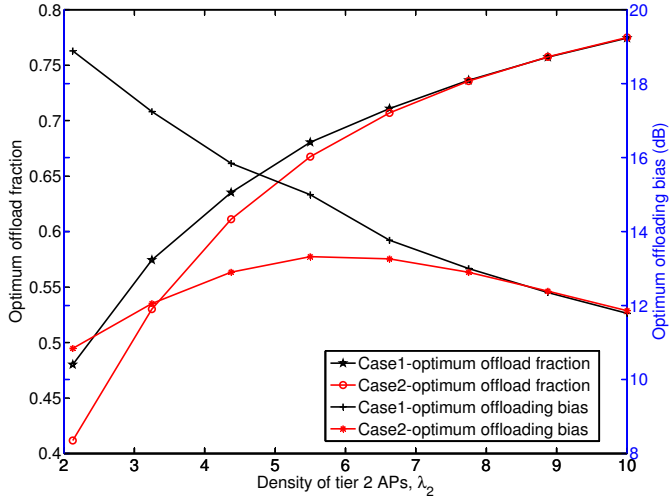}
		\caption{ Effect  of  backhaul bandwidth and small cell density on  the  optimum  association  bias  and  optimum traffic offload fraction.}
	\label{fig:bias_trafficfraction_trend}
\end{figure}
\section{Conclusion}\label{sec:conclusion}
We have provided an analytical framework for offloading and resource partitioning in  co-channel  heterogeneous networks. With the heterogeneity  in base stations becoming prominent in emerging wireless networks these two techniques are set to play an important role in radio resource management. To the best of our knowledge, this is first work to derive rate distribution in a heterogeneous cellular network, while incorporating resource partitioning and limited  bandwidth backhauls. The availability  of a functional form for rate as a function of system parameters opens a plethora of avenues to  gain design insights. Using the developed analysis, the importance of combining load balancing with  resource partitioning was clearly established. It was further shown that the rate is a key metric for studying these techniques and insights based on just $\SINR$ are  inconclusive.  Investigating the impact of offloading on  rate distribution in the uplink of HetNets could be an interesting future research direction.
\appendices
\section{}\label{sec:proofassocpr}
\begin{IEEEproof}[Proof of Lemma \ref{lem:aspr}]
Using the definition of the three disjoint sets, the respective association probabilities are
\begin{align*}\label{eq:passoc1}
 \passoc_{1} & = \pr\left(\power{1}\NDIST{1}^{-\ple{1}}>\power{2}\nbias{2}\NDIST{2}^{-\ple{2}}\right)\\&= \int_{z>0}\pr\left(\NDIST{2}>(\npower{2}\nbias{2})^{1/\ple{2}}\ndistns^{1/\nple{2}}\right)f_{\NDIST{1}}(\ndistns)\mathrm{d} \ndistns
 \end{align*}
\small
\begin{equation}
\begin{aligned}\label{eq:passoc1}
  \passoc_{\ipueindex} & = \pr\left(\power{2}\NDIST{2}^{-\ple{2}}>\power{1}\NDIST{1}^{-\ple{1}}\right)\\&=
\int_{z>0}\pr\left(\NDIST{1}>(\npower{1})^{1/\ple{1}}\ndistns^{1/\nple{1}}\right)f_{\NDIST{2}}(\ndistns)\mathrm{d} \ndistns \\
\displaystyle\passoc_{\opueindex }&=\pr\left(\power{2}\bias{2}\NDIST{2}^{-\ple{2}}>\power{1}\NDIST{1}^{-\ple{1}}\bigcap\power{2}\NDIST{2}^{-\ple{2}}<\power{1}\NDIST{1}^{-\ple{1}}\right)\\
\hspace{-0.05in}=\int_{z>0}\pr&\left(\left(\frac{\npower{1}}{\bias{2}}\right)^{1/\ple{1}}\ndistns^{1/\nple{1}}\leq\NDIST{1}<(\npower{1})^{1/\ple{1}}\ndistns^{1/\nple{1}}\right)f_{\NDIST{2}}(\ndistns)\mathrm{d} \ndistns.
\end{aligned} 
\end{equation}
\normalsize
Now
\begin{equation}\label{eq:nullpr}
\pr\left(\NDIST{k} >\ndistns \right) =\pr\left(\PPP{k} \cap b(0,\ndistns) = \emptyset \right) = \exp\left(-\pi \dnsty{k}\ndistns^2\right),
\end{equation}
 where $b(0,z)$ is the Euclidean ball of radius $\ndistns$ centered at origin.
The probability distribution function (PDF) $f_{\NDIST{k}}(\ndistns)$ can then be written as
\small
\begin{equation}\label{eq:pdfndist}
f_{\NDIST{k}}(\ndistns) = \frac{\mathrm{d}}{\mathrm{d}\ndistns}\{1- \pr(\NDIST{k}>\ndistns)\} = 2\pi\dnsty{k}\ndistns\exp(-\pi\dnsty{k}\ndistns^2),\,\, \forall \ndistns \geq 0.
\end{equation}
\normalsize
Using  (\ref{eq:nullpr}) and (\ref{eq:pdfndist}) in (\ref{eq:passoc1}) gives Lemma \ref{lem:aspr}.
\end{IEEEproof}

\section{}\label{sec:proofoloadpgf}
\begin{IEEEproof}[Proof of Lemma \ref{lem:oloadpgf}]
 Assuming that the  palm inversion formula  \cite{moller}, which relates the area of Poisson Voronoi (PV) containing the origin to that of a typical PV, holds for  multiplicatively weighted PV \footnote{Formal proof is out of the scope of this paper and will be a dealt in a future publication.}, the distribution of the association area  of the tagged AP of tier-$k$, $\area_{k}^{'}$  can be written as 
\begin{equation}
f_{\area_{k}^{'}}(\areans) \propto \areans f_{\area_{k}}(\areans).
\end{equation}
For a typical user $u \in \set{U}_l$, the load at the tagged AP $\load{l}$ comprises of the typical user and other users $\oload{l}$ (say) associated with the AP. Using Remark \ref{rem:associatioarea},  
the PMF of  load 
\begin{multline}
\pmf_{l}(n) =\pr\left(\load{l}=n\right)= \int_{\areans>0}\exp(\userdnsty \areans)\frac{(\userdnsty \areans)^{(n-1)}}{(n-1)!}f_{\area_{l}^{'}}(\areans)\mathrm{d\areans}\\
=\frac{3.5^{3.5}}{(n-1)!}\frac{\Gamma(n+3.5)}{\Gamma(3.5)}\left(\frac{\userdnsty\passoc_{l}}{\dnsty{\tmap{l}}}\right)^{n-1}\left(3.5 + \frac{\userdnsty\passoc_{l}}{\dnsty{\tmap{l}}}\right)^{-(n+3.5)} \\n\geq 1.
\end{multline}
A similar approach was taken for single-tier setting in \cite{YuKim13}  and for  multi-tier setting in \cite{SinDhiAnd13}.
\end{IEEEproof}


\section{}\label{sec:proofpcov}
\begin{IEEEproof}[Proof of Lemma \ref{lem:pcov}]
In this proof we first derive the distribution of the distance between the typical user $u$ and the tagged AP  when $u\in \set{U}_l$. Let $\NDISTc{l}$ denote this distance, then 
\begin{align}
\pr(\NDISTc{l}>\ndistnsc) & 
=\pr\left(\NDIST{\tmap{l}}>\ndistnsc| u \in \set{U}_l \right)=\frac{\pr\left(\NDIST{\tmap{l}}>\ndistnsc, u \in \set{U}_l\right)}{\pr\left(u \in \set{U}_l\right)}.\label{eq:derivndist1}
\end{align}
Using the proof of Lemma \ref{lem:aspr} in Appendix \ref{sec:proofassocpr}, the corresponding PDFs are 
\begin{equation}
\begin{aligned}\label{eq:ndist4}
f_{\NDISTc{1}}(\ndistnsc) &= \frac{2\pi\dnsty{1}}{\passoc_{1}} \ndistnsc\exp\left(-\pi\sum_{k=1}^{2}\dnsty{k}(\npower{k}\nbias{k})^{2/\ple{k}}\ndistnsc^{2/\nple{k}}\right)\\
f_{\NDISTc{\ipueindex}}(\ndistnsc)&  = \frac{2\pi\dnsty{2}}{\passoc_{\ipueindex}} \ndistnsc\exp\left(-\pi\sum_{k=1}^{2}\dnsty{k}(\npower{k})^{2/\ple{k}}\ndistnsc^{2/\nple{k}}\right)\\
f_{\NDISTc{\opueindex }}(\ndistnsc) & = \frac{2\pi\dnsty{2}}{\passoc_{\opueindex }} \ndistnsc\exp\left(-\pi\sum_{k=1}^{2}\dnsty{k}\npower{k}^{2/\ple{k}}\ndistnsc^{2/\nple{k}}\right)\\&\left\{\exp\left(-\pi\sum_{k=1}^{2}\dnsty{k}\npower{k}^{2/\ple{k}}\ndistnsc^{2/\nple{k}}(\nbias{k}^{2/\ple{k}}-1)\right)-1\right\}.
\end{aligned}
\end{equation}
Conditioned on serving AP being $s_l$, the  Laplace transform  of interference can be expressed as the Laplace functional of $\tierPPP{k}$
\begin{align}
\mgf_{I_{y,k}}(s)&=\cexpect{\exp(-sI_{y,k})}{I_{y,k}}\\&=\expect{\exp\left(-s\power{k}\sum_{x\in \tierPPP{k} \setminus s_l} \ap_{xk} \chanl_{x} x^{-\ple{k}}\right)} \\
&\overset{(a)}{=} \cexpect{\prod_{x\in\tierPPP{k}\setminus s_l}\mgf_{\chanl_x}\left(s\power{k}x^{-\ple{k}}\right)}{\tierPPP{k}}\\
&\overset{(b)}{=}\exp\left(-2\pi\dnsty{k}\int_{\ndist{kl}(y)}^{\infty}\left\{1-\mgf_{\chanl_x}\left(s\power{k}x^{-\ple{k}}\right)\right\}x\mathrm{d}x\right)\\
&\overset{(c)}{=}\exp\left(-2\pi\dnsty{k}\int_{\ndist{kl}(y)}^{\infty}\frac{x}{1+(s\power{k})^{-1} x^{\ple{k}}}\mathrm{d}x\right),
\end{align}
where (a) follows from the independence of $\chanl_x$, (b) is obtained using the PGFL  \cite{mecke_book} of $\tierPPP{k}$, and (c) follows by using the MGF of an exponential RV with unit  mean.
In the above expression,  $\ndist{kl}(y)$ is the lower bound on distance of the closest interferer in $\uth{k}$ tier, which can be obtained by using (\ref{eq:association}) as
\begin{equation}
\begin{aligned}
\text{if }& l=1: \ndist{21}(y)=(\npower{2}\nbias{2})^{1/\ple{2}}\ndistnsc^{\ple{1}/\ple{2}}, \ndist{11}(y)= \ndistnsc\\
\text{if }& l=\ipueindex: \ndist{1\ipueindex}(y)=(\npower{1})^{1/\ple{2}}\ndistnsc^{\ple{1}/\ple{2}}, \ndist{2\ipueindex}(y) = \ndistnsc\\
\text{if }& l=\opueindex : \ndist{2\opueindex}(y) =\ndistnsc
\end{aligned}
\end{equation}
Using change of variables  with $t= (s\power{k})^{-2/\ple{k}}x^2$, the integral can be simplified as
\small
\begin{align}
\int_{\ndist{kl}(y)}^{\infty}\frac{2x}{1+(s\power{k})^{-1} x^{\ple{k}}}\mathrm{d}x &= (s\power{k})^{2/\ple{k}}\int_{(s\power{k})^{-2/\ple{k}}\ndist{kl}(y)^2}^{\infty}\frac{\mathrm{d}t}{1+t^{\ple{k}/2}}\nonumber\\
&= (s\power{k})^{2/\ple{k}}\Z\left(1,\ple{k},\frac{\ndist{kl}(y)^\ple{k}}{s\power{k}}\right),\end{align}
\normalsize
giving the  Laplace transform of interference 
\begin{equation}\label{eq:mgfinterference}
\mgf_{I_{y,k}}(s)= \exp\left(-\pi\dnsty{k}(s\power{k})^{2/\ple{k}}\Z\left(1,\ple{k},\frac{\ndist{kl}(y)^\ple{k}}{s\power{k}}\right)\right),
\end{equation}
where \[\Z(a,b,c)= a^{2/b}\int_{(\frac{c}{a})^{2/b}}^\infty \frac{\mathrm{d} u}{ 1+ u^{b/2}}.\]
The $\SINR$ coverage of  user $u\in\uset_l$ is
\begin{equation}\label{eq:pcovproof0}
\pcov_{l}(\SINRthresh) = \int_{\ndistnsc\geq0} \pr(\SINR > \SINRthresh|u\in \uset_l,\NDISTc{l}=\ndistnsc) f_{\NDISTc{l}}(\ndistnsc) \mathrm{d} \ndistnsc.
\end{equation}
Using  the $\SINR$ expression in (\ref{eq:sinrdef})
\begin{align}
 &\pr(\SINR>\SINRthresh|u \in \mue,\NDISTc{l}=\ndistnsc)=\pr\left(\frac{\power{1}\chanl_{\ndistnsc} \ndistnsc^{-\ple{1}}}{\sum_{k=1}^2 I_{y,k} + \noisepower}>\SINRthresh\right)\\&= \pr\left(\chanl_\ndistnsc > \ndistnsc^{\ple{1}}{\power{1}}^{-1}\SINRthresh\left\{\sum_{k=1}^2 I_{y,k} +\noisepower \right\}\right)\\
&=\expect{\exp\left(-\ndistnsc^{\ple{1}}\SINRthresh\power{1}^{-1}\left\{\sum_{k=1}^2I_{y,k} + \noisepower\right\}\right)}\\
&\overset{(a)}{=}\exp\left(-\frac{\SINRthresh}{\SNR_{1}(\ndistnsc)}\right)\prod_{k=1}^2 \cexpect{\exp\left(-\ndistnsc^{\ple{1}}\SINRthresh\power{1}^{-1}I_{y,k}\right)}{I_{y,k}}\\
&=\exp\left(-\frac{\SINRthresh}{\SNR_{1}(\ndistnsc)}\right)\prod_{k=1}^2\mgf_{I_{y,k}}\left(\ndistnsc^{\ple{1}}\SINRthresh\power{1}^{-1}\right) \label{eq:pcovproof1},
\end{align}
where $\SNR_{1}(\ndistnsc) = \frac{\power{1}\ndistnsc^{-\ple{1}}}{\noisepower}$ and (a) follows from the independence of $I_{y,k}$. Similarly 
\begin{multline}
 \pr(\SINR>\SINRthresh | u\in\ipue, \NDISTc{l}=\ndistnsc)\\=\exp\left(-\frac{\SINRthresh}{\SNR_{2}(\ndistnsc)}\right)\prod_{k=1}^2\mgf_{I_{y,k}}\left(\ndistnsc^{\ple{2}}\SINRthresh\power{2}^{-1}\right) \label{eq:pcovproof2},
\end{multline}
and 
\begin{multline}
 \pr(\SINR(\ndistnsc)>\SINRthresh| u \in \opue,\NDISTc{l}=\ndistnsc)\\=\exp\left(-\frac{\SINRthresh}{\SNR_{2}(\ndistnsc)}\right)\mgf_{I_{y,2}}\left(\ndistnsc^{\ple{2}}\SINRthresh\power{2}^{-1}\right). \label{eq:pcovproof3}
\end{multline}
Using the PDF distribution  (\ref{eq:ndist4}) in  (\ref{eq:pcovproof0}) along with  (\ref{eq:pcovproof1})-(\ref{eq:pcovproof3})   and (\ref{eq:mgfinterference}), the $\SINR$ coverage expressions given in Lemma \ref{lem:pcov} are obtained.
The overall $\SINR$ coverage of a typical user is then obtained  using the law of total probability to get
$ \pcov(\SINRthresh) = \sum_{l}\pcov_{l}(\SINRthresh) \passoc_{l}. $
\end{IEEEproof}

\section{}\label{sec:proof_pcovcomp}
\begin{IEEEproof}[Proof of claims]
\textbf{Claim 1:}
The partial derivative of $\pcov^w$ with respect to offloading bias $b$ is 
\small
\begin{multline*}
 -a\sqrt{p}\frac{-\frac{\SINRthresh}{b+{\SINRthresh}}\frac{1}{2\sqrt{b}} + \frac{1}{2\sqrt{b}}}{\left\{\sqrt{\SINRthresh}\tan^{-1}(\sqrt{\SINRthresh})+1 + a\sqrt{p}(\sqrt{\SINRthresh}\tan^{-1}(\sqrt{\SINRthresh/b})+ \sqrt{b})\right\}^2}  \\
 - \frac{1}{a\sqrt{p}}\frac{\frac{\SINRthresh}{1+\SINRthresh b} - \frac{1}{2b^{3/2}}}{\left\{\sqrt{\SINRthresh}\tan^{-1}(\sqrt{\SINRthresh})+1 + \frac{1}{a\sqrt{p}}(\sqrt{\SINRthresh}\tan^{-1}(\sqrt{b\SINRthresh})+ \sqrt{1/b})\right\}^2}.
\end{multline*}
\normalsize
Since $b, \SINRthresh\geq 0$, hence $-\frac{\SINRthresh}{b+{\SINRthresh}}\frac{1}{2\sqrt{b}} + \frac{1}{2\sqrt{b}} \geq 0$. Also, if \begin{align}
&\SINRthresh \geq 1 \implies \SINRthresh \geq 1/b \text{ for } b\geq 1\implies \SINRthresh  \geq \frac{1}{2b^{3/2}-b}\\
&\implies \frac{\SINRthresh}{1+\SINRthresh b} - \frac{1}{2b^{3/2}} \geq 0 \implies \nabla_b\pcov^w \leq 0.
\end{align}
Thus, for $\SINRthresh \geq 1$  the $\SINR$ coverage decreases for all $b \geq 1 $.\\
\textbf{Claim 2:}
Approximating $\tan^{-1}(a)\approx a$ and substituting $x$ for $\sqrt{b}$, the partial derivative of coverage with respect to $x$ is 
\begin{align}
\nabla_x\pcov = & \nabla_x\Bigg\{\frac{1}{\sqrt{\SINRthresh}\tan^{-1}(\sqrt{\SINRthresh})+1 + a\sqrt{p}(\frac{{\SINRthresh}}{x}+ x)} \nonumber\\& +
 \frac{1}{\sqrt{\SINRthresh}\tan^{-1}(\sqrt{\SINRthresh})+1 + \frac{1}{a\sqrt{p}x}}\Bigg\} \\
= & \frac{a\sqrt{p}(\frac{\SINRthresh}{x^2}-1)}{\left\{v + a\sqrt{p}(\frac{{\SINRthresh}}{x}+ x)\right\}^2}+ \frac{1}{a\sqrt{p}x^2(v+ \frac{1}{a\sqrt{p}x})^2},
\end{align}
 where $v \triangleq \sqrt{\SINRthresh}\tan^{-1}(\sqrt{\SINRthresh})+1$.
The roots of the equation $\nabla_x\pcov =0$ are the zeros of the polynomial 
\begin{multline}
P(x) = x^4 a^2p(v^2-1) + 2x^3a\sqrt{p}v(1-\SINRthresh)\\ -x^2\left\{v^2-1 + a^2p\SINRthresh(v^2+2)\right\}- 4xa\sqrt{p}v\SINRthresh-a^2p\SINRthresh^2-\SINRthresh.
\end{multline}
Since $v>1$, using the Descartes sign rule the polynomial $P(x)$ has $1$ positive root and upto $3$ negative roots. 
The value of the positive root can be be upper bounded \cite{Stefanescujucs05} by 
\small
\begin{multline}
U  = \max\Bigg\{\left[3\frac{ v^2-1 + a^2p\SINRthresh(v^2+2)}{a^2p(v^2-1)}\right]^{1/2},\\\left[3\frac{4a\sqrt{p}v}{a^2p(v^2-1)}\right]^{1/3}, \left[3\frac{a^2p\SINRthresh^2+\SINRthresh}{a^2p(v^2-1)}\right]^{1/4}\Bigg\} \text{ if }\SINRthresh\leq1 \end{multline}\begin{multline}
U  = \max\Bigg\{\left[4\frac{2a\sqrt{p}v(\SINRthresh-1)}{a^2p(v^2-1)}\right],\left[4\frac{ v^2-1 + a^2p\SINRthresh(v^2+2)}{a^2p(v^2-1)}\right]^{1/2},\\\left[4\frac{4a\sqrt{p}v}{a^2p(v^2-1)}\right]^{1/3}, \left[4\frac{a^2p\SINRthresh^2+\SINRthresh}{a^2p(v^2-1)}\right]^{1/4}\Bigg\} \text{ if }\SINRthresh>1.
\end{multline}
\normalsize
 Further, the upper bound on the positive roots of $P(-x)$ is given by 
 \small
\begin{multline}
L  = \max\Bigg\{\left[2\frac{ v^2-1 + a^2p\SINRthresh(v^2+2)}{a^2p(v^2-1)}\right]^{1/2},\\ \left[2\frac{a^2p\SINRthresh^2+\SINRthresh}{a^2p(v^2-1)}\right]^{1/4}\Bigg\} \text{ if }\SINRthresh>1
\end{multline}
\begin{multline}
L  = \max\Bigg\{\left[3\frac{2a\sqrt{p}v(1-\SINRthresh)}{a^2p(v^2-1)}\right],\left[3\frac{ v^2-1 + a^2p\SINRthresh(v^2+2)}{a^2p(v^2-1)}\right]^{1/2},\\ \left[3\frac{a^2p\SINRthresh^2+\SINRthresh}{a^2p(v^2-1)}\right]^{1/4}\Bigg\} \text{ if }\SINRthresh\leq1.
\end{multline}
\normalsize
Note that $-L$ is the lower bound on the negative roots of $P(x)$, since they are same as the positive roots of $P(-x)$. Clearly, both  $U$ and $L$ are  inversely proportional to the density of small cells $a$. Since
\begin{align}
-L\leq \sqrt{b} \leq U \implies b \leq \max\left\{U^2, L^2\right\},
\end{align}
therefore the upper bound on optimal bias is inversely proportional to the density of small cells $a$.\\
\textbf{Claim 3:} The $\SINR$ coverage at very large offloading bias is 
\begin{align}
&\pcov|_{b=\infty} \triangleq \lim_{b\to\infty} \pcov \\&=
 \frac{1}{v + \frac{v}{a\sqrt{p}}} + \frac{1}{v}-\frac{1}{v+\frac{1}{a\sqrt{p}}},
\end{align}
where $v \triangleq \sqrt{\SINRthresh}\tan^{-1}(\sqrt{\SINRthresh})+1$. With the knowledge that $\pcov|_{b=1} = \pcov^w|_{b=1}$ we get 
\begin{align*}
\pcov|_{b=1} - \pcov|_{b=\infty} = & \frac{1}{v+av} -\frac{1}{v} + \frac{1}{v+\frac{1}{a}}
= \frac{v^2 -v }{v(v+\frac{1}{a})(v+av)}\nonumber\\
& > 0 \text{ since } v >1. 
\end{align*}
Thus, $\pcov|_{b=\infty} < \pcov|_{b=1}$.
\end{IEEEproof}

\bibliographystyle{ieeetr}
\bibliography{IEEEabrv,/Users/singh/Dropbox/research/refoffload}


\end{document}